\documentclass[journal]{IEEEtran}

\usepackage{cite}
\usepackage[dvips]{graphicx}
\usepackage{subfigure}
\usepackage{amsthm}
\newtheorem{theorem}{Theorem}
\usepackage{amsmath}
\newtheorem{lemma}{Lemma}
\usepackage{algorithm}
\usepackage{algorithmic}
\usepackage{enumerate}
\usepackage{flushend}
\usepackage{booktabs}
\usepackage{url}
\usepackage{xcolor}
\newcommand{\ws}[1]{\textcolor{black}{#1}}
\usepackage{verbatim}

\begin{document}
\title{A Localization Method Avoiding Flip Ambiguities for micro-UAVs with Bounded Distance Measurement Errors}

\author{
\IEEEauthorblockN{Qingbei~Guo\IEEEauthorrefmark{1}\IEEEauthorrefmark{2},
                  Yuan~Zhang\IEEEauthorrefmark{1},~\IEEEmembership{Senior Member,~IEEE},
                  Jaime Lloret\IEEEauthorrefmark{3},~\IEEEmembership{Senior Member,~IEEE},
                  Burak Kantarci\IEEEauthorrefmark{4},~\IEEEmembership{Senior Member,~IEEE},
                  Winston K.G. Seah\IEEEauthorrefmark{5},~\IEEEmembership{Senior Member,~IEEE}}

\IEEEauthorblockA{\IEEEauthorrefmark{1}Shandong Provincial Key Laboratory of Network based Intelligent Computing, University of Jinan, Jinan, Shandong, China\\
Email: \{ise\_guoqb, yzhang\}@ujn.edu.cn\\
\IEEEauthorrefmark{1}Corresponding author}

\IEEEauthorblockA{\IEEEauthorrefmark{2}Jiangsu Provincial Engineering Laboratory of Pattern Recognition and Computational Intelligence, Jiangnan University, Wuxi 214122, China}

\IEEEauthorblockA{\IEEEauthorrefmark{3}Integrated Management Coastal Research Institute, Universidad Polit¡äecnica de Valencia, Spain}

\IEEEauthorblockA{\IEEEauthorrefmark{4}School of Electrical Engineering and Computer Science, University of Ottawa, Canada}

\IEEEauthorblockA{\IEEEauthorrefmark{5}School of Engineering and Computer Science Victoria University of Wellington, New Zealand}
\thanks{IEEE Transactions on Mobile Computing, Submitted Apr 2, 2017 Accepted Aug 7, 2018}}

\markboth{Decision on TMC-2017-04-0233.R2: Accept as regular paper}%
{Shell \MakeLowercase{\textit{et al.}}: Bare Demo of IEEEtran.cls for IEEE Journals}

\maketitle

\begin{abstract}
Localization is a fundamental function in cooperative control of micro unmanned aerial vehicles (UAVs), but is easily affected by flip ambiguities because of measurement errors and flying motions. This study proposes a localization method that can avoid the occurrence of flip ambiguities in bounded distance measurement errors \ws{and constrained flying motions}; to demonstrate its efficacy, the method is implemented on bilateration and trilateration. For bilateration, an improved bi-boundary model based on the unit disk graph model is created to compensate for the shortage of distance constraints, and two boundaries are estimated as the communication range constraint. The characteristic of the intersections of the communication range and distance constraints is studied to present a unique localization criterion which can avoid the occurrence of flip ambiguities. Similarly, for trilateration, another unique localization criterion for avoiding flip ambiguities is proposed according to the characteristic of the intersections of three distance constraints. The theoretical proof shows that these proposed criteria are correct. A localization algorithm is constructed based on these two criteria. The algorithm is validated using simulations for different scenarios and parameters, and the proposed method is shown to provide excellent localization performance in terms of average estimated error. \ws{Our code can be found at: 
\url{https://github.com/QingbeiGuo/AFALA.git}}.
\end{abstract}


\IEEEpeerreviewmaketitle

\section{Introduction}
\IEEEPARstart{T}{}he use of multiple unmanned aerial vehicles (UAVs) has become very popular in civil and military applications. As a key technical problem, localization of multiple mirco-UAVs finds its application in location tracking, formation flight and cooperative mission, etc. As the most frequently used localization method, GPS suffers from large location errors of 10-30m on average. Therefore, different techniques have been proposed to address the localization problem \cite{Youssef05,Cheng04,Patwari03,Potdar09}. Localization techniques, which are based on distance measurement (e.g., trilateration and bilateration), have always attracted significant research interest \cite{Carter81,Rappaport96,Bernhardt87,Lee02,Gezici05}. However, flip ambiguity (FA) is one of the major problems of these localization techniques, especially in the presence of measurement errors \cite{Severi09,Wang11,Liu14,Yang16,Zhang12_2}. The measurement error between the measurement distance and the true distance always occur because of external environment noise \cite{Carter81,Rappaport96,Bernhardt87,Lee02,Gezici05}, that easily causes FA and results in an incorrect estimate \cite{Niculescu04} \cite{Moore04} \cite{Kannan10}. Moreover, once a FA has occurred, it not only affects the current localization but also causes erroneous results for subsequent localizations.

\ws{Global rigidity is widely used for localization to ensure that the result is unique \cite{Zhang12_1,Aspnes06,Anderson10}.} A network can be uniquely located if and only if its corresponding graph is globally rigid \cite{Eren04}. However, even if the network is globally rigid in the presence of errors, FA may still occur. In trilateration, the three measured distances may locate their connected node to a false side, which causes FA due to errors in the measurements \cite{Akcan13,Kannan08}.

Bilateration only requires network rigidity conditions to be met, not global rigidity \cite{Fang06,Goldenberg06,Yang09}, and reduces the reliance on high node density. However, FA in bilateration deserves more attention. Given that only two distances exist, the localization conditions are so insufficient that additional localization constraints are required. The communication range constraint of the nodes is widely adopted as an additional localization constraint and therefore has an important role in determining the final location from the candidate locations. More importantly, an incorrect choice inevitably results in the occurrence of FA so that the estimation of the communication range of a node becomes very important.

The unit disk graph \cite{Olivia15,Kaewprapha11,Kuhn04} is commonly used to model the communication range of a wireless node. Each node has a single circle with a radius equal to a fixed value, and two nodes are connected if the distance between them is below the specified threshold. However, in reality, the boundary between reachable and unreachable areas cannot be clearly defined most of the time. Communication range between nodes is impeded by external conditions, for instance, buildings in a city \cite{Iwashige15}, and internal factors, such as its energy availability especially after operating for long duration. Therefore, the fixed communication range assumption has its limitation in localization.

Our proposed localization method is intended to address the aforementioned problems. \ws{In our method, we assume the distance measurement errors to be bounded \cite{Shi17} and the flying motions to be constrained \cite{Hu04}. An improved bi-boundary model of the communication range is proposed as localization constraints based on the model of unit disk graph. The new model depends on only distances and connectivity, and calculates the double boundaries of communication ranges through the bounded measurement error.} Given that both bilateration and trilateration are analyzed, every constraint is regarded as a possible location region. The intersections of two distance constraints form two possible localization regions first, then the third constraint (bi-boundary and distance constraint for bilateration and trilateration, respectively) is used to eliminate one region, which causes FA, by analyzing the characteristic of their intersections. Therefore, the remaining region must contain the true location and the estimated location without the possibility of FA.

Accordingly, the main contributions of this paper are as follows:
\begin{itemize}
\item An improved bi-boundary model based on the unit disk graph model is presented by analyzing the connectivity characteristics of a wireless node.
\item Based on the double constraints of distance and bi-boundary, a localization criterion that avoids FA in bilateration is proposed.
\item Another localization criterion that avoids FA in trilateration is also developed.
\item A localization algorithm based on the above two localization criteria, which dramatically improves the location accuracy, is constructed and evaluated through extensive simulations.
\end{itemize}

The rest of this paper is organized as follows:
Section \ref{sec:RelatedWork} introduces the related work. Section \ref{sec:ProblemFormulation} formulates the localization problem. Section \ref{sec:LocalizationCriteriaAlgorithm} presents an improved bi-boundary model based on the unit disk graph model, describes the localization criteria that avoid FA in bilateration and trilateration, and provides their proofs. Based on these localization criteria, the localization algorithm is developed. In section \ref{sec:PerformanceValidation}, the localization algorithm is validated through comprehensive simulation. Finally, section \ref{sec:Conclusion} presents the conclusion.

\section{Related Work}\label{sec:RelatedWork}

Various works have investigated the localization and the phenomenon of FA that may hamper the unique localization of distributed nodes, and have proposed different localization methods for reducing or avoiding FA. These methods can be classified into two categories: methods based on global and non-global rigidity properties which take trilateration and bilateration as representatives, respectively. Further details of these methods are provided below.

A unique and anchor-free localization algorithm, which also resolves the FA problem in its second step, has been proposed by Zhang \textit{et al.} \cite{Zhang12_2}. A novel combination of distance and direction estimation technique is introduced to detect and estimate ranges between neighbors. Using both distance and angle information, we construct unidirectional local coordinate systems to avoid the reflection ambiguity. However, angle measurement is too expensive to be imbedded in most nodes, leading to few practical applications and inspiring our development of a new localization algorithm to solve the FA issue.

Trilateration is the most widely used localization method based on distance measurement \cite{Eren04,Aspnes06}. The estimated location of a node is determined using measured distances to three other nodes that have known locations and are not collinear. Trilateration without measurement errors is uniquely localizable; however, trilateration with measurement errors tends to suffer from the effect of FA \cite{Kannan08,Evrendilek11}.

To prevent incorrect localization caused by FA, a robust quadrilateral was used to perform trilateration in \cite{Moore04,Savarese02,Sittile08,Hendrickson92}. Moore \textit{et al.}\cite{Moore04} introduced the notion of robust quadrilateral, which is a fully connected quadrilateral with global rigid properties, and used it to reduce the probability of FA. They found all possible robust quadrilaterals and then realized global localization by overlapping any two robust quadrilaterals that have three common nodes. The simulation results showed that the aforementioned method is suitable for the localization of high node density, because it depends on these robust quadrilaterals that require a high node density to meet the feature of global rigidity.

Kannan \textit{et al.} \cite{Kannan10} pointed out that a node may be estimated at a flipped location caused by measurement errors in a globally rigid graph when its three neighboring nodes with known location are nearly collinear. They proposed a robust criterion based on the robust quadrilateral to calculate the probability of FA, to eliminate all the locations that might have caused it, and to improve localization performance. The simulations show that, compared with the robust quadrilateral, this method decreases not only the average estimation error but also the average number of localized nodes resulting from the robust criterion, which requires more constraints than the robust quadrilateral. Although more constraints can avoid problems of FA, they also limit the location of more nodes.

Akcan \textit{et al.} \cite{Akcan13} proposed a heuristic solution based on the notion of a ``safe-triangle" to mitigate the problem of FA in a trilateration network with range noise. First proposed here, a safe-triangle is formed by three nodes with known location, where the distance of any node to the line passing through two other nodes is larger than a set threshold. If a triangle is not a safe-triangle, its nodes are unable to provide trilateration for other nodes. The main aim is to minimize the number of FA. The simulation results show that the algorithm can achieve better performance than trilateration. However, similar to trilateration, the safe-triangle algorithm also requires a high node density.

To reduce localization dependence on global rigidity, many algorithms \cite{Fang06,Goldenberg06,Olivia15} are proposed for network localization using bilateration which requires rigid conditions but not necessarily globally rigid. Based on \cite{Fang06}, Goldenberg \textit{et al.} \cite{Goldenberg06} proposed the ``Sweeps" algorithm for network localization using bilateration. The Sweeps algorithm achieves one of the best performance in terms of number of localized nodes using bilateration. However, the algorithm requires a bilateration ordering to exist in the network, a condition which does not always exist in a sparse network; thus, it may fail in many localizable networks \cite{Goldenberg06,Yang09}.

Oliva \textit{et al.} \cite{Olivia15} introduced a model of shadow edges to extend trilateration to bilateration, which can produce effective solutions despite the lack of localization condition. The algorithm shows better performance than trilateration, and is able to localize the network even when trilateration fails. However, this method is based on unit disk graphs and only limited to noiseless environments.

The objective of our research is to avoid the occurrence of FA and achieve the accurate localization network using trilateration and bilateration with the errors including bounded measurement errors and constrained motions. A localization algorithm is proposed based on two classes of localization criteria that avoid the FA problem. In our algorithm, the two classes of localization criteria are derived by gradually analyzing the intersection characteristic of all the constrains. Then, trilateration and bilateration are combined to take advantage of the best of them, thus making them applicable in various kinds of networks that are sparse or not.

\section{Problem Formulation}\label{sec:ProblemFormulation}
\ws{Without loss of generality, the localization problem of a micro-UAV network modelled as a mobile sensor network can be formulated in 2D space \cite{Zhang10,Baggio08}.} Multiple micro-UAVs, which consists of a set of $n$ nodes, denoted as $s_1$ to $s_n$, are flying in a physical region. Each node has a limited communication range and establishes a wireless link with a neighboring node, which is called its neighboring node only if they are within the communication range of each other. A node is assumed to be capable of estimating the distance to the neighboring node using distance measurement technology. Given the constraints on energy consumption and implementation environment, most nodes do not know their locations except anchors which can obtain their own locations by using GPS. In this study, the micro-UAV network is analyzed without anchor nodes. The communication range of node $s_i$ is denoted by $r_i$, and its estimated communication range by $\hat{r}_i$, its true location by $p_i$, and its estimated location by $\hat{p}_i$. The true distance between any two neighboring nodes $s_i$ and $s_j$ is denoted by $d_{ij}$, and the measured distance by $\hat{d}_{ij}$, where $i$, $j$ = 1,2,\ldots, $n$.

\ws{For a connected network, we assume that the distances between neighboring nodes can be acquired in a time unit \cite{Keung10}. The speed of a node $s_i$ is the distance travelled by the node per time unit \cite{Hu04}, denoted as $d_i$, and $d_{max}$ is the maximum distances travelled by these nodes in a time unit. Therefore, in each time unit, the micro-UAV network can be described by a model of undirected graph $G_t$ = ($V$, $E$) with a nonempty vertex set $V$ = \{1,2,\ldots, $n$\} and edge set $E$, where each vertex $i \in V$ uniquely represents a node $s_i$, and each edge $e(i,j)$ $\in$ $E$ is uniquely associated with a node pair ($s_i$, $s_j$), for which $s_i$ and $s_j$ are neighbors, and $\hat{d}_{ij}$ is known. The measurement error of $e(i,j)$ is denoted by $e_{ij}$ such that $e_{ij} = d_{ij} - \hat{d}_{ij}$, and $e_{max}$ is the maximum measurement errors.} \ws{The network topology can and will change in the different time units because of the measurement errors and the flying trajectories. However, during the localization process}, $s_i$ has a possible localization region set $R_i$ = \{$r_{i1}$, $r_{i2}$\} and corresponding candidate location set $P_i$ = \{$p_{i1}$, $p_{i2}$\} because of the FA phenomenon. While the true location $p_i$ is difficult to be determined due to the errors consisting of the measurement errors and the flying distances, $R_i$ and $P_i$ can be determined by the measured distances and the errors. FA occurs when the estimated location and the true location are not in the same localization region. Therefore, to determine the true localization region becomes a key problem in deciding whether FA occurs. By gradually analyzing the characteristic of the intersections corresponding to all the constraints, our method eliminates the localization region in which each location causes FA instead of choosing which one contains the true location. Therefore, the remaining region is identified as the true localization region, and the corresponding candidate location is regarded as the estimated location with no possibility of FA.

\section{Localization Criteria and Localization Algorithm}\label{sec:LocalizationCriteriaAlgorithm}

\subsection{Bi-boundary Model of Communication Range Constraint}

In this section, an improved bi-boundary model based on the unit disk graph is presented to model the topology of wireless sensor networks with the errors which consist of the bounded measurement errors and the constrained motions. Now, we first focus on the effect of only the measurement errors, and the flying motions will be introduced in Section \ref{sec:LocalizationCriteriaAlgorithm}.D. In this bi-boundary model, each node has double concentric circles with two different radii, and the space between two circles represents \ws{an uncertainty in the communication range}. The inner and outer boundaries represent the lower limit of the reachable range and the upper limit of the unreachable range, \ws{respectively}; thus, this bi-boundary design is more suitable for the actual environment. The calculation of the two boundaries is independent of the measuring technique and instead completely depends on the knowledge of the distances, the measurement errors and connectivity between nodes, without maintaining the individual parameter of the communication range. The process of estimating the upper and lower boundaries of a wireless node is given in detail below.


\ws{Let $\varepsilon > 0$ be a threshold of the distance variation, such that $|e_{ij}| = |d_{ij} - \hat{d}_{ij}| \leq \varepsilon = e_{max}$.} Thus,
\begin{equation}
\label{eqn:01}
  d_{ij} \in [\hat{d}_{ij} - \varepsilon, \hat{d}_{ij} + \varepsilon].
\end{equation}

\noindent
Following the aforementioned analysis on the neighboring node and communication noise, for every node $s_i$, the estimated communication range $\hat{r}_i$ can be bounded as
\begin{equation}
\label{eqn:02}
  \hat{r}_i \geq max(D_i^1) - \varepsilon,
\end{equation}
where $D_i^k$ represents the measured distance set between the nodes that are at least \emph{k}-hop away from node $s_i$, where $k$ = 1,2,\ldots,$n$. In particular, $D_i^1$ = \{$\hat{d}_{ij}$ $|$ $e(i,j)$ $\in$ $E$\}, $D_i^2$ = \{$\hat{d}_{il}$+$\hat{d}_{lj}$ $|$ $e(i,j)$ $\notin$ $E$, $e(i,j)$ $\in$ $E$, $e(i,j)$ $\in$ $E$, and $l \neq i$, $l \neq j$\}.

\begin{figure}
  \centering
  \includegraphics[trim=0mm 5mm 0mm 5mm, width=2.02in]{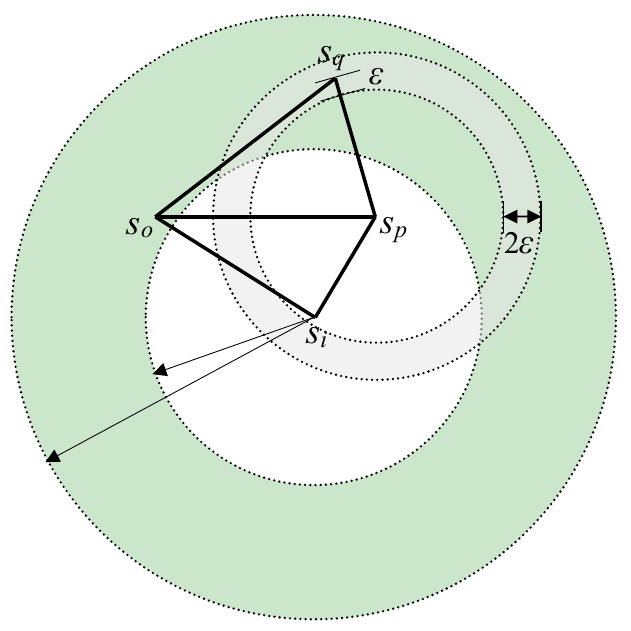}
  \caption{Example  of estimating the communication range. For node $s_i$, the minimum of the estimated communication range is $\hat{d}_{io} - \varepsilon$, and the maximum is $\hat{d}_{pi} + \hat{d}_{pq} + 2\varepsilon$.}
  \label{fig:1}
\end{figure}

To further analyze $r_i$, the scenario is illustrated in Fig. \ref{fig:1}. Given the unknown location of $s_i$ which is connected with both $s_o$ and $s_p$ but is not connected with $s_q$, $\hat{d}_{iq}$ cannot be directly calculated using $p_i$ and $p_q$; thus, it cannot contribute to the estimation of $r_i$. The sum of any two sides of a triangle is greater than the third side, so that $r_i$ is less than the sum of any $2$-hop distance from $s_i$ to $s_q$. Therefore, considering the factor of the measurement errors, the following condition must be satisfied:
\begin{equation}
\label{eqn:03}
  \hat{r}_i < min(D_i^2) + 2\varepsilon.
\end{equation}

\noindent
Finally, the estimated communication range $\hat{r}_i$, which satisfies equations (\ref{eqn:02}) and (\ref{eqn:03}), is bounded as follows:
\begin{equation}
\label{eqn:04}
  max(D_i^1) - \varepsilon \leq \hat{r}_i < min(D_i^2) + 2\varepsilon.
\end{equation}

\noindent
The connection property of a node therefore affects its estimated result. Moreover, higher node density, which increases the degree of node connectivity, makes the estimate more precise.

\subsection{FA Avoidance Criteria for Localization using Bilateration}

\begin{figure*}
\begin{minipage}{1\textwidth}
  \centering
  \subfigure[Noraml case]{
  \includegraphics[trim=0mm 0mm 0mm 0mm, width=1.66in]{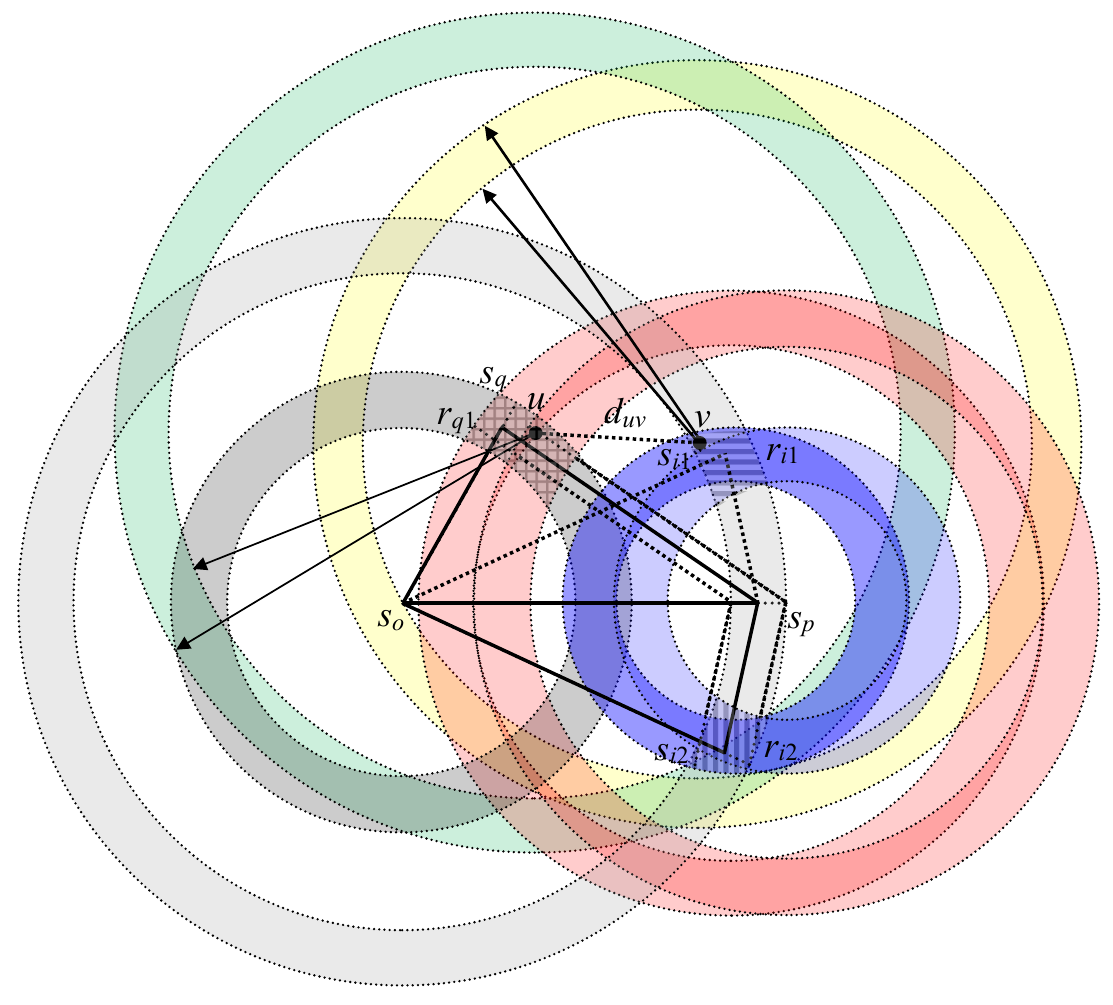}}
  \hspace{0in}
  \subfigure[Near-collinear case]{
  \includegraphics[trim=0mm 0mm 0mm 0mm, width=1.66in]{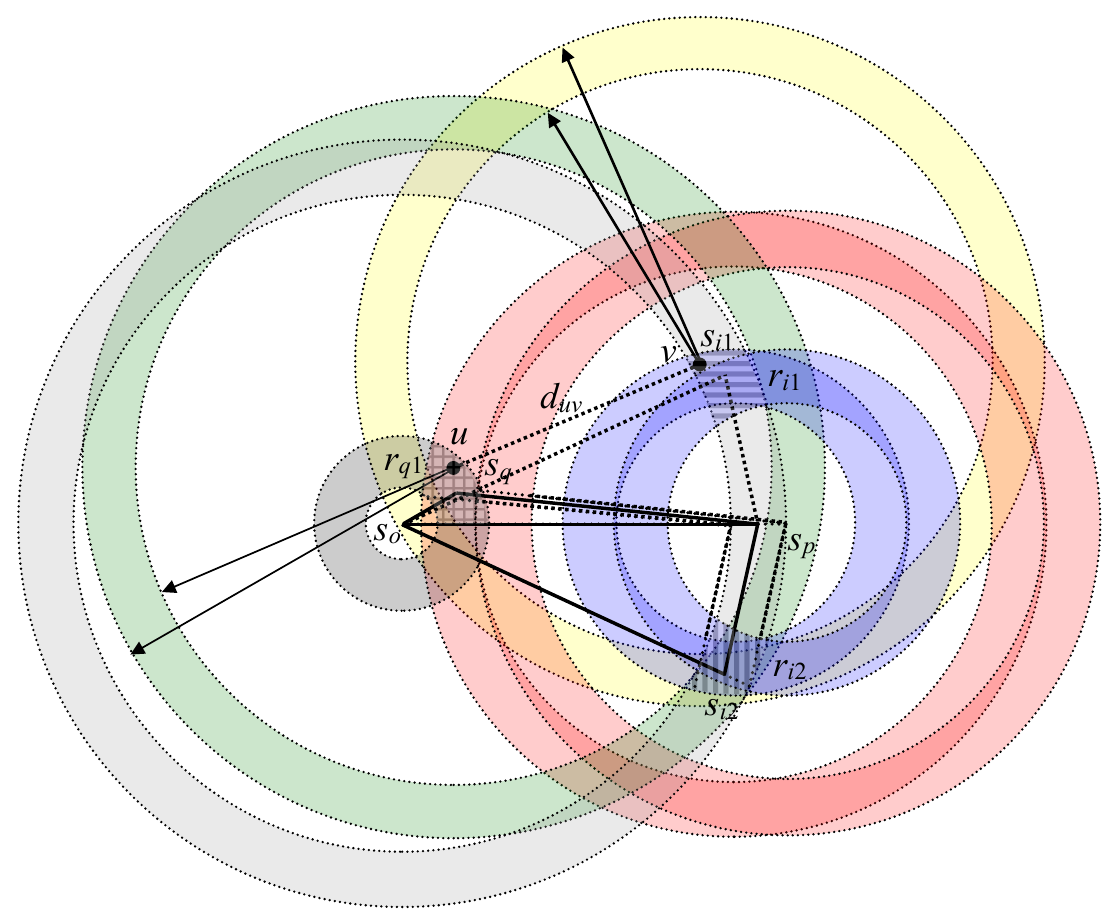}}
  \hspace{0in}
  \subfigure[Overlapping case]{
  \includegraphics[trim=0mm 0mm 0mm 0mm, width=1.66in]{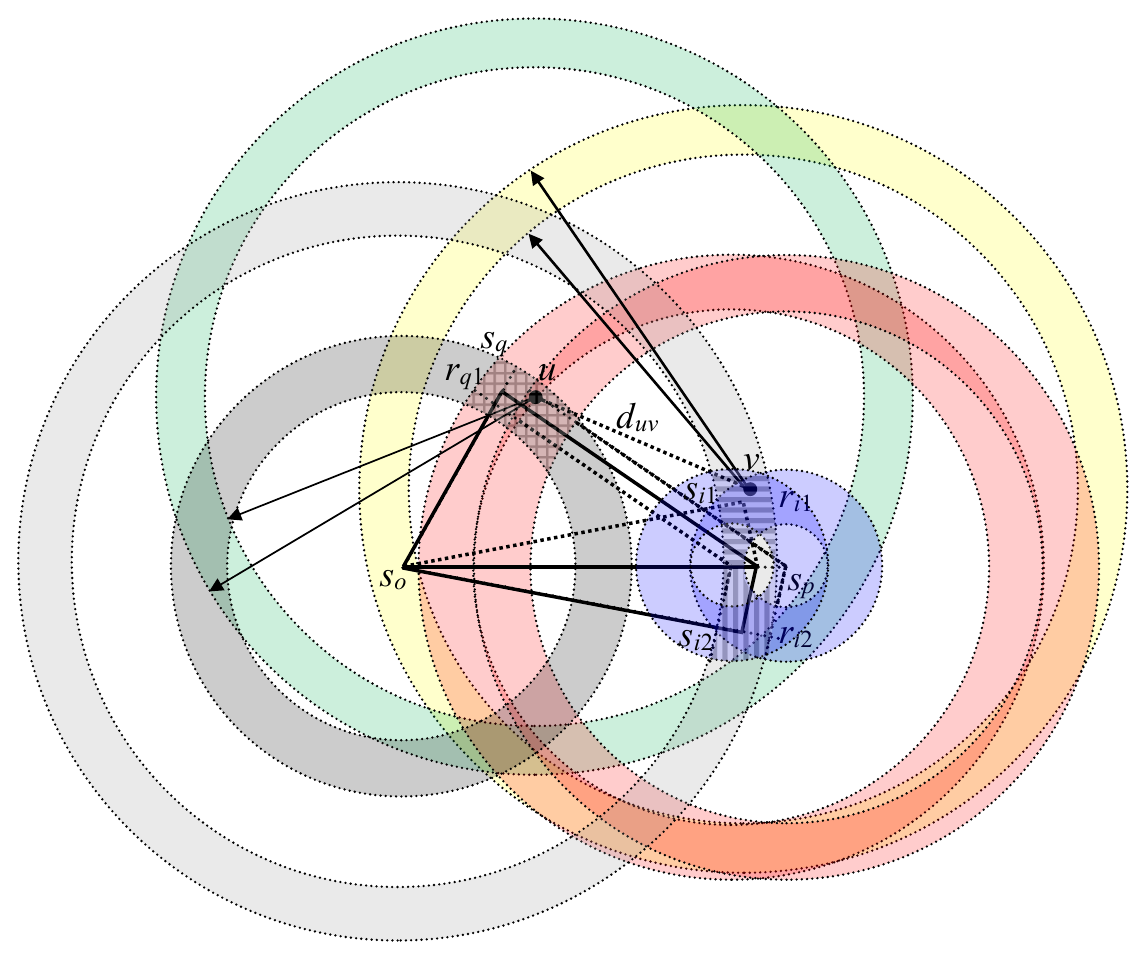}}
  \hspace{0in}
  \subfigure[Both near-collinear case and overlapping case]{
  \includegraphics[trim=0mm 0mm 0mm 0mm, width=1.66in]{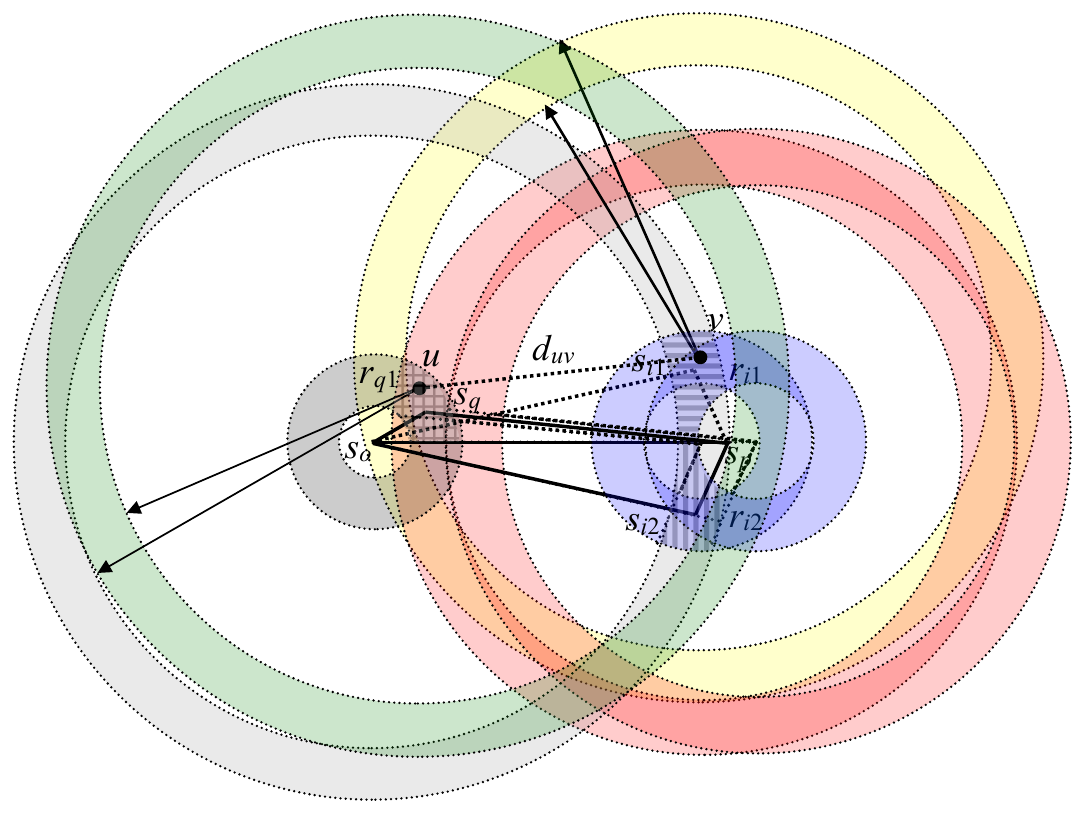}}
  {\caption{Localization without flip ambiguities using bilateration.}
  \label{fig:2}}
\end{minipage}
\end{figure*}

In this section, we introduce the scenario of using bilateration and present a localization criteria for avoiding FA using the bi-boundary and distance constraints. Fig. \ref{fig:2} represents the corresponding graph of localization without flip ambiguities using bilateration. Fig. \ref{fig:2} (a) shows a normal case without near-collinear and overlapping occurrences. Fig. \ref{fig:2} (b) and (c) depict a near-collinear case and an overlapping case, respectively. Moreover, a compositive case with near-collinear and overlapping occurrences is shown in Fig. \ref{fig:2} (d). For each case above, we will follow the same process because our proposed localization method adapts to all the cases. The location of node $s_o$ is assumed to be the origin, and another node $s_p$ is selected to form the positive $x$-axis, and the third node $s_q$ is located in the upper half-plane. The intersections of the black annulus and the red annuli form the upper region $r_{q1}$, which is the possible localization region of $s_q$. The possible localization regions of $s_i$ are the two regions $r_{i1}$ and $r_{i2}$ consisting of the intersections of the gray annulus and the blue annuli. The three regions of $r_{r1}$, $r_{r2}$ and $r_{q1}$ are highlighted 
with horizontal stripes, vertical stripes and mesh, respectively. $u$ and $v$ are any points of $r_{q1}$ and $r_{i1}$, respectively, and their communication ranges are denoted by the green annulus and the yellow annulus, respectively. Obviously, both $u$ and $v$ are within the communication range of each other, and hence they must be adjacent. Furthermore, each point of $r_{i1}$ clearly is within the communication range of $u$; thus, for $u$, the whole region $r_{i1}$ is not a possible localization region of $s_i$. Note that due to the uncertainty of the true location $p_p$, there exist countless red annuli and blue annuli whose centers are at every point in between $(\hat{d}_{op} - \varepsilon, 0)$ and $(\hat{d}_{op} + \varepsilon, 0)$, but only two red annuli and two blue annuli at both ends are shown to simplify the figure.

The true distances $d_{op}$, $d_{oq}$ and $d_{pq}$ among the three nodes can be bounded as follows:
\begin{equation}
\label{eqn:05}
  d_{op} \in [\hat{d}_{op} - \varepsilon, \hat{d}_{op} + \varepsilon],
\end{equation}
\begin{equation}
\label{eqn:06}
  d_{oq} \in [\hat{d}_{oq} - \varepsilon, \hat{d}_{oq} + \varepsilon],
\end{equation}
\begin{equation}
\label{eqn:07}
  d_{pq} \in [\hat{d}_{pq} - \varepsilon, \hat{d}_{pq} + \varepsilon].
\end{equation}

\noindent
It is obvious that the true location $p_p$ of $s_p$ is between $(\hat{d}_{op} - \varepsilon, 0)$ and $(\hat{d}_{op} + \varepsilon, 0)$. The true location $p_q$ of node $s_q$ is calculated based on $p_o$, $p_p$, $d_{oq}$ and $d_{pq}$. Based on equations (\ref{eqn:06}) and (\ref{eqn:07}), the two close intervals form two annulus regions, which can be defined using $\hat{d}_{oq}$ and $\hat{d}_{pq}$ as
\begin{equation}
\begin{aligned}
\label{eqn:08}
  r_{oq}& = \{(x,y) | \\
  &(\hat{d}_{oq} - \varepsilon)^2 \leq (x-x_o)^2 + (y-y_o)^2 \leq (\hat{d}_{oq} + \varepsilon)^2\},
\end{aligned}
\end{equation}
\begin{equation}
\begin{aligned}
\label{eqn:09}
  r_{pq}& = \{(x,y) | \\
  &(\hat{d}_{pq} - \varepsilon)^2 \leq (x-x_p)^2 + (y-y_p)^2 \leq (\hat{d}_{pq} + \varepsilon)^2\},
\end{aligned}
\end{equation}
where $(x_o,y_o)$ and $(x_p,y_p)$ are the location coordinates of $s_o$ and $s_p$, respectively. Since $s_o$ is assumed to be the origin, equation (\ref{eqn:08}) can be also written as follows:
\begin{equation}
\label{eqn:10}
  r_{oq} = \{(x,y) | (\hat{d}_{oq} - \varepsilon)^2 \leq x^2 + y^2 \leq (\hat{d}_{oq} + \varepsilon)^2\}.
\end{equation}

\noindent
However, as equation (\ref{eqn:05}) shows, the true location $p_p$ is uncertain. Thus, equation (\ref{eqn:09}) can be also written as follows:
\begin{equation}
\label{eqn:11}
\begin{aligned}
  r_{pq} = \{&(x,y)~| \\
   ((\hat{d}_{pq} & - \varepsilon)^2 \leq (x - \hat{d}_{op} + \varepsilon)^2 + y^2 \leq (\hat{d}_{pq} + \varepsilon)^2) \vee \\
   ((\hat{d}_{pq} & - \varepsilon)^2 \leq (x - \hat{d}_{op} - \varepsilon)^2 + y^2 \leq (\hat{d}_{pq} + \varepsilon)^2)\}.
\end{aligned}
\end{equation}

\noindent
The intersections of $r_{oq}$ and $r_{pq}$ form the two regions $r_{q1}$ and $r_{q2}$. The two regions are certainly symmetrical with respect to the edge $e(o,p)$, but they can be adjacent or nonadjacent in various situations \cite{Kannan10}. To ensure that they are separated, the edge $e(o,p)$ is added as the boundary between them. The half-plane that contains $s_q$ is denoted by $H$, and the complementary half-plane is denoted by $H'$. Hence, $H$ and $H'$ can be written as follows:
\begin{equation}
\label{eqn:12}
  H = \{(x,y)|y > 0\},
\end{equation}
\begin{equation}
\label{eqn:13}
  H' = \{(x,y)|y \leq 0\}.
\end{equation}

\noindent
Thus, the possible localization region of node $s_q$ can be defined as follows:
\begin{equation}
\label{eqn:14}
  r_{q1} = \{(x,y) | (x,y) \in r_{oq} \cap r_{pq} \cap H\}.
\end{equation}

Similarly, this method can also be extended to the analysis of the possible localization regions of $s_i$. The two annulus regions related with $s_i$ can be defined as follows:
\begin{equation}
\label{eqn:15}
  r_{oi} = \{(x,y) | (\hat{d}_{oi} - \varepsilon)^2 \leq x^2 + y^2 \leq (\hat{d}_{oi} + \varepsilon)^2\},
\end{equation}
\begin{equation}
\label{eqn:16}
\begin{aligned}
  r_{pi} = \{&(x,y) | \\
  ((\hat{d}_{pi} & - \varepsilon)^2 \leq (x - \hat{d}_{op} + \varepsilon)^2 + y^2 \leq (\hat{d}_{pi} + \varepsilon)^2) \vee \\
  ((\hat{d}_{pi} & - \varepsilon)^2 \leq (x - \hat{d}_{op} - \varepsilon)^2 + y^2 \leq (\hat{d}_{pi} + \varepsilon)^2)\}.
\end{aligned}
\end{equation}

\noindent
Thus, the two possible localization regions of $s_i$ can be defined as follows:
\begin{equation}
\label{eqn:17}
  r_{i1} = \{(x,y) | (x,y) \in r_{oi} \cap r_{pi} \cap H\},
\end{equation}
\begin{equation}
\label{eqn:18}
  r_{i2} = \{(x,y) | (x,y) \in r_{oi} \cap r_{pi} \cap H'\}.
\end{equation}

\noindent
The true location of node $s_i$ is in one of the two regions $r_{i1}$ or $r_{i2}$. To obtain an unambiguous estimated location, additional knowledge is required to determine which region contains the true location. In this paper, the communication ranges of the nodes $s_q$ and $s_i$ are used to address the ambiguity problem. According to the aforementioned bi-boundary model, the communication ranges of $s_q$ and $s_i$ can be bounded as follows:
\begin{equation}
\label{eqn:19}
  r_q \in [max(D_q^1) - \varepsilon, min(D_q^2) + 2\varepsilon],
\end{equation}
\begin{equation}
\label{eqn:20}
  r_i \in [max(D_i^1) - \varepsilon, min(D_i^2) + 2\varepsilon].
\end{equation}

\noindent
Hence, the maximum communication distance, within which $s_q$ and $s_i$ can certainly communicate with each other, is defined by:
\begin{equation}
\label{eqn:21}
  D_{qi} = min(max(D_q^1) - \varepsilon, max(D_i^1) - \varepsilon).
\end{equation}

Lastly, Theorem 1 proves whether node $s_i$ can be uniquely localizable in $\triangle opq$.

\begin{theorem}
\label{thm:01}
Given a localized triangle $\triangle opq$, where its three localized nodes $s_o$, $s_p$ and $s_q$ are adjacent to each other, and an unknown node $s_i$, such that the edges $e(o,i)$ and $e(p,i)$ exist, but $e(q,i)$ does not exist. \ws{Let $\varepsilon = e_{max} > 0$ be a threshold of the distance measurement errors.} $s_i$ can be a unique localization, and its sufficient conditions are as follows:
\begin{enumerate}[\indent(1)]
\item $\forall x\in r_{q1}$, $\forall y\in r_{i1}$, $d_{xy} < D_{qi}$,
\item $\exists u\in r_{q1}$, $\exists v\in r_{i2}$, $d_{uv} > D_{qi}$.
\end{enumerate}
\end{theorem}
\begin{proof}
$r_{i1}$ and $r_{i2}$ are two possible localization regions of $s_i$ based on constraints of the two distance $\hat{d}_{oi}$ and $\hat{d}_{pi}$. If every pair of points between $r_{q1}$ and $r_{i1}$ are adjacent while at least one pair of points between $r_{q1}$ and $r_{i2}$ are beyond $D_{qi}$, then the whole region $r_{i1}$ is excluded as a possible localization region of $s_i$. Therefore, the region $r_{i2}$ is the only possible region which contains the true location $p_i$, and no possibility exists for the occurrence of FA. In this case, both nodes $s_i$ and $s_q$ are clearly on the opposite of the edge $e(o,p)$.
\end{proof}

\subsection{FA Avoidance Criteria for Localization using Trilateration}

\begin{figure*}
\begin{minipage}{1\textwidth}
  \centering
  \subfigure[Noraml case]{
  \includegraphics[trim=0mm 0mm 0mm 0mm, width=1.66in]{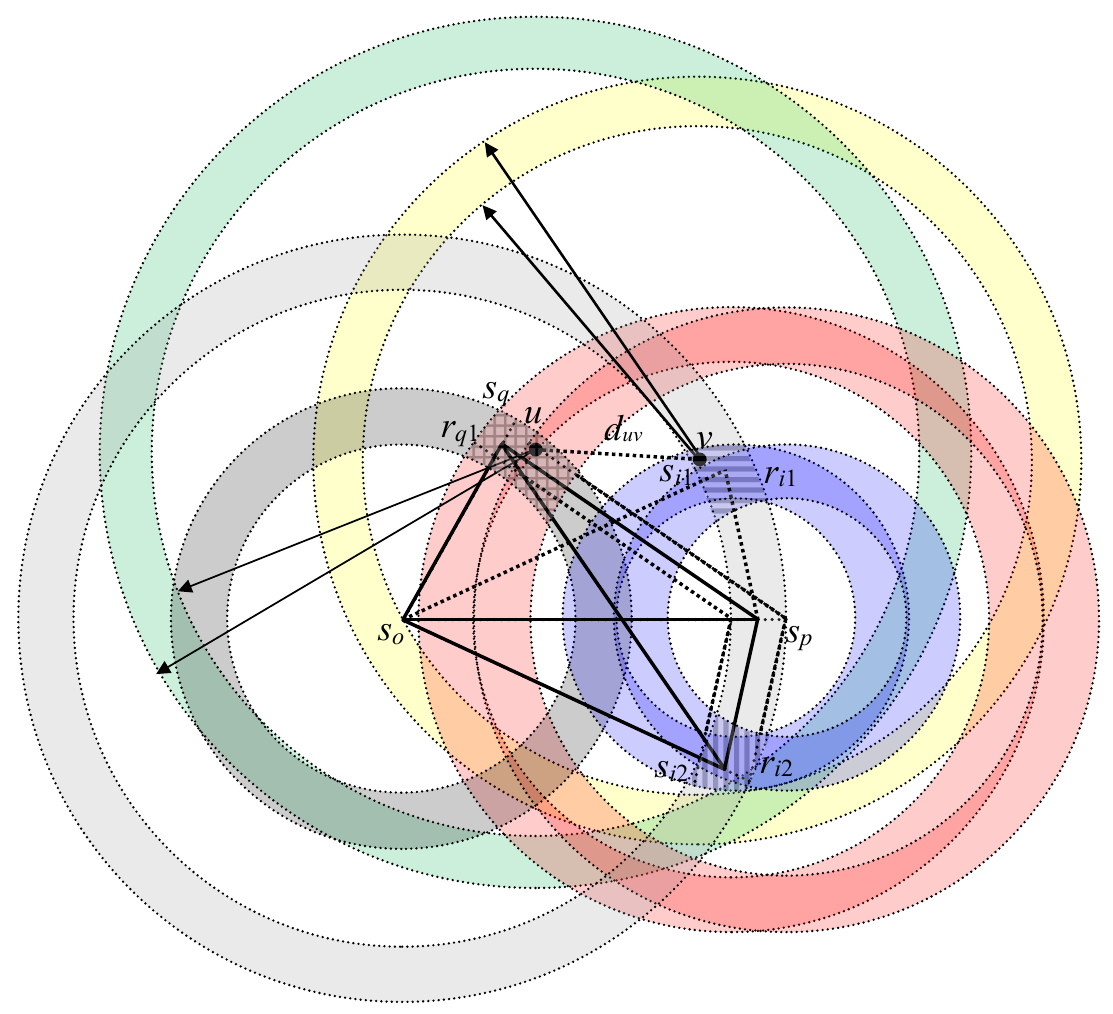}}
  \hspace{0in}
  \subfigure[Near-collinear case]{
  \includegraphics[trim=0mm 0mm 0mm 0mm, width=1.66in]{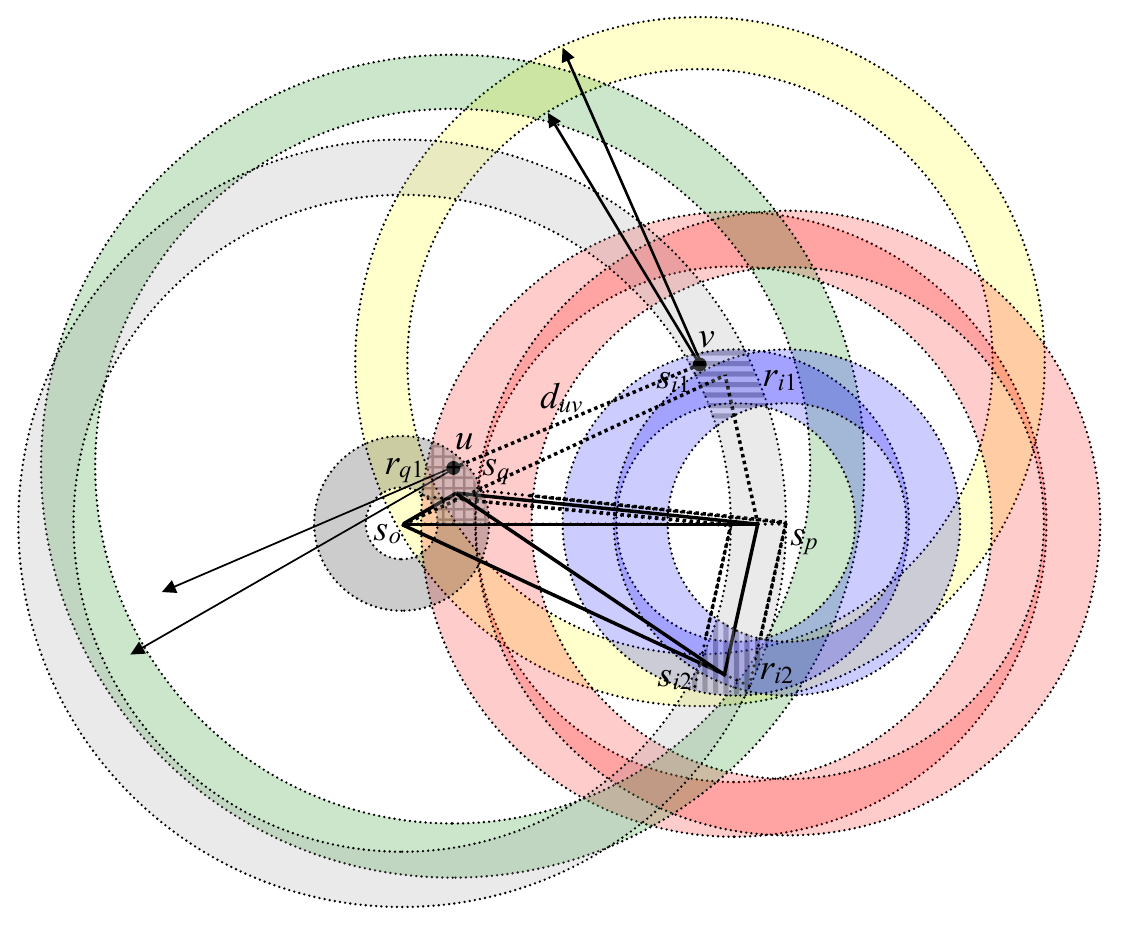}}
  \hspace{0in}
  \subfigure[Overlapping case]{
  \includegraphics[trim=0mm 0mm 0mm 0mm, width=1.66in]{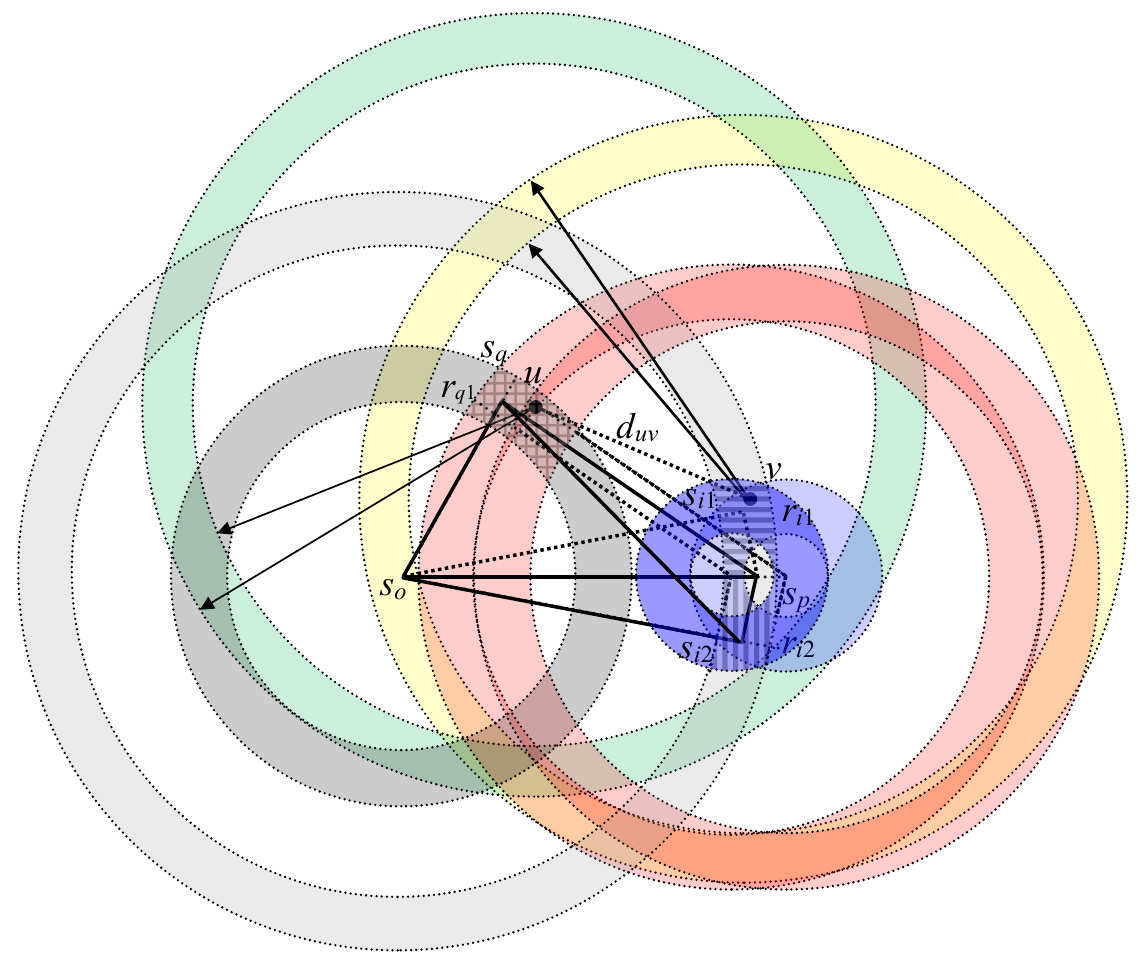}}
  \hspace{0in}
  \subfigure[Both near-collinear case and overlapping case]{
  \includegraphics[trim=0mm 0mm 0mm 0mm, width=1.66in]{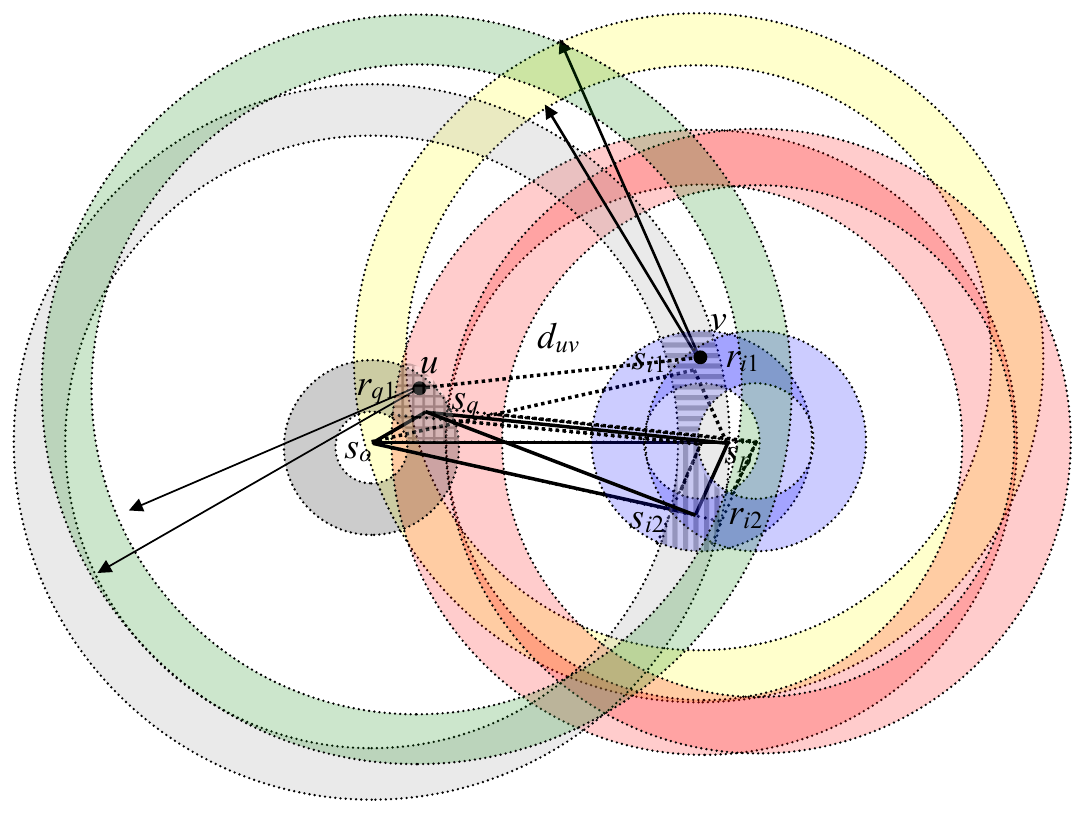}}
  {\caption{Localization without flip ambiguities using trilateration.}
  \label{fig:3}}
\end{minipage}
\end{figure*}

We now address the scenario where trilateration is used without flip ambiguities, and Fig. \ref{fig:3} shows the corresponding graph of this scenario. Just like Fig. \ref{fig:2}, a normal case is depicted without near-collinear and overlapping occurrences in Fig. \ref{fig:3} (a). Fig. \ref{fig:3} (b) and (c) show a near-collinear case and an overlapping case, respectively. And then, Fig. \ref{fig:3} (d) shows a compositive case with near-collinear and overlapping occurrences. Due to the feasibility of our method to all the cases, we will also follow the same process. The detection of the FA occurrences is now analyzed. The possible localization region of $s_q$ is the upper region $r_{q1}$ consisting of the intersections of the black annulus and the red annuli. The intersections of the gray annulus and the blue annuli form two regions $r_{i1}$ and $r_{i2}$, which are the possible localization regions of $s_i$. The three regions highlighted 
with horizontal stripes, vertical stripes and mesh indicate $r_{r1}$, $r_{r2}$ and $r_{q1}$, respectively. $u$ and $v$ are any points of $r_{q1}$ and $r_{i1}$, respectively. 
Since $d_{uv} \notin [\hat{d}_{qi} - \varepsilon, \hat{d}_{qi} + \varepsilon]$, $v$ is without the constraints of the distance $\hat{d}_{qi}$, and each point of $r_{i1}$ is clearly without the constraint of $\hat{d}_{qi}$. Therefore, for $u$, the whole region $r_{i1}$ is not a possible localization region of $s_i$.

Similar to the analysis for the scenario using bilateration, the two annuli $r_{oq}$ and $r_{pq}$ are also respectively defined as follows:
\begin{equation}
\label{eqn:22}
  r_{oq} = \{(x,y) | (\hat{d}_{oq} - \varepsilon)^2 \leq x^2 + y^2 \leq (\hat{d}_{oq} + \varepsilon)^2\},
\end{equation}
\begin{equation}
\label{eqn:23}
\begin{aligned}
  r_{pq} = \{&(x,y) | \\
  ((\hat{d}_{pq} & - \varepsilon)^2 \leq (x - \hat{d}_{op} + \varepsilon)^2 + y^2 \leq (\hat{d}_{pq} + \varepsilon)^2) \vee\\
                     ((\hat{d}_{pq} & - \varepsilon)^2 \leq (x - \hat{d}_{op} - \varepsilon)^2 + y^2 \leq (\hat{d}_{pq} + \varepsilon)^2)\}.
\end{aligned}
\end{equation}

\noindent
Thus, the possible localization region of node $s_q$ is also defined as follows:
\begin{equation}
\label{eqn:24}
  r_{q1} = \{(x,y) | (x,y) \in r_{oq} \cap r_{pq} \cap H\}.
\end{equation}

Similarly, using the two measured distances $\hat{d}_{oi}$ and $\hat{d}_{pi}$ that are affected by noises, the two annuli $r_{oi}$ and $r_{pi}$ are also respectively defined as follows:
\begin{equation}
\label{eqn:25}
  r_{oi} = \{(x,y) | (\hat{d}_{oi} - \varepsilon)^2 \leq x^2 + y^2 \leq (\hat{d}_{oi} + \varepsilon)^2\},
\end{equation}
\begin{equation}
\label{eqn:26}
\begin{aligned}
  r_{pi} = \{&(x,y) | \\
  ((\hat{d}_{pi} & - \varepsilon)^2 \leq (x - \hat{d}_{op} + \varepsilon)^2 + y^2 \leq (\hat{d}_{pi} + \varepsilon)^2) \vee \\
  ((\hat{d}_{pi} & - \varepsilon)^2 \leq (x - \hat{d}_{op} - \varepsilon)^2 + y^2 \leq (\hat{d}_{pi} + \varepsilon)^2)\}.\hspace*{-2ex}
\end{aligned}
\end{equation}

\noindent
Following which, the two possible localization regions of $s_i$ are also defined as follows:
\begin{equation}
\label{eqn:27}
  r_{i1} = \{(x,y) | (x,y) \in r_{oi} \cap r_{pi} \cap H\},
\end{equation}
\begin{equation}
\label{eqn:28}
  r_{i2} = \{(x,y) | (x,y) \in r_{oi} \cap r_{pi} \cap H'\}.
\end{equation}

\noindent
The true location $p_i$ must be located in one of the two separated regions $r_{i1}$ and $r_{i2}$, but incorrect calculation can cause FA; thus, to address the problem of FA and derive the estimated location, we use a third annulus $r_{qi}$, which is defined by the following equation:
\begin{equation}
\label{eqn:29}
  r_{qi} = \{(x,y) | (\hat{d}_{qi} - \varepsilon)^2 \leq (x-x_q)^2 + (y-y_q)^2 \leq (\hat{d}_{qi} + \varepsilon)^2\}.
\end{equation}

To summarize, we have the following theorem:

\begin{theorem}
\label{thm:02}
Given a localized triangle $\triangle opq$, where its three localized nodes $s_o$, $s_p$ and $s_q$ are adjacent to each other, and an unknown node $s_i$, such that the edges $e(o,i)$, $e(p,i)$ and $e(q,i)$ exist. \ws{Let $\varepsilon = e_{max} > 0$ be a threshold of the distance measurement errors.} $s_i$ can be a unique localization, and its sufficient conditions are as follows:
\begin{enumerate}[\indent(1)]
\item $\forall x\in r_{q1}$, $\forall y\in r_{i1}$, $d_{xy} \notin [\hat{d}_{qi} - \varepsilon, \hat{d}_{qi} + \varepsilon]$,
\item $\exists u\in r_{q1}$, $\exists v\in r_{i2}$, $d_{uv} \in [\hat{d}_{qi} - \varepsilon, \hat{d}_{qi} + \varepsilon]$.
\end{enumerate}
Or,
\begin{enumerate}[\indent(1)]
\item $\forall x\in r_{q1}$, $\forall y\in r_{i2}$, $d_{xy} \notin [\hat{d}_{qi} - \varepsilon, \hat{d}_{qi} + \varepsilon]$,
\item $\exists u\in r_{q1}$, $\exists v\in r_{i1}$, $d_{uv} \in [\hat{d}_{qi} - \varepsilon, \hat{d}_{qi} + \varepsilon]$.
\end{enumerate}
\end{theorem}
\begin{proof}
Following the same analysis in Theorem \ref{thm:01}, $r_{i1}$ and $r_{i2}$ are two possible localization regions of $s_i$ based on constraints of the two distances $\hat{d}_{oi}$ and $\hat{d}_{pi}$. There are at least one pair of points, whose distances are under the constraints of the third distance $\hat{d}_{qi}$, between $r_{q1}$ and $r_{i2}$ (or $r_{i1}$), whereas there is no any pair of points, whose distance is without the constraints of the third distance $\hat{d}_{qi}$, between $r_{q1}$ and $r_{i1}$ (or $r_{i2}$). Thus, the whole region $r_{i1}$ (or $r_{i2}$) is eliminated from the possible localization regions of $s_i$, and the other region $r_{i2}$ (or $r_{i1}$) is a unique estimated region that causes no FA.
\end{proof}

\begin{figure*}
\begin{minipage}{1\textwidth}
  \centering
  \subfigure[A random network graph]{
  \includegraphics[trim=0mm 0mm 0mm 0mm, width=1.66in]{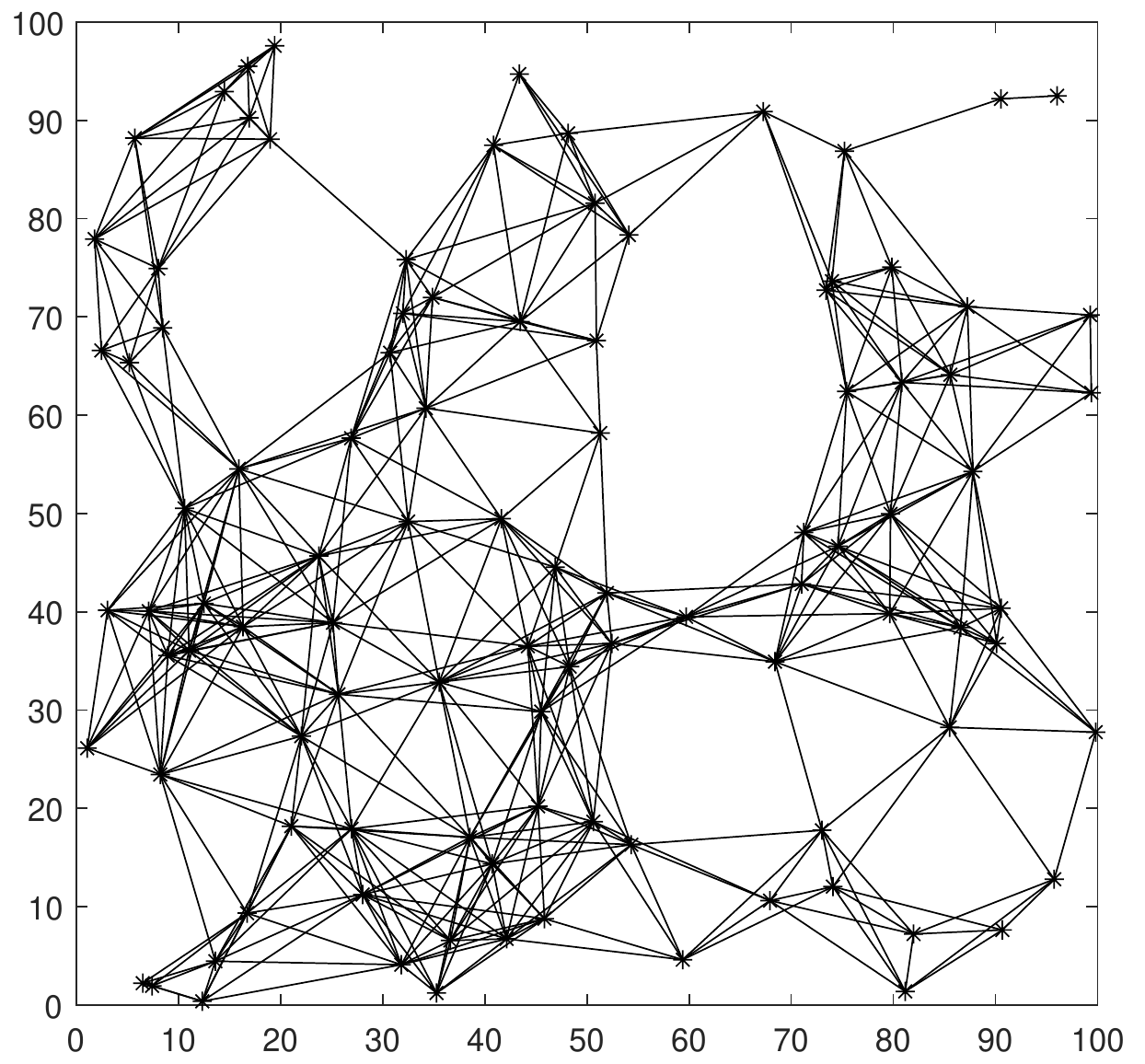}}
  \hspace{0in}
  \subfigure[TLA]{
  \includegraphics[trim=0mm 0mm 0mm 0mm, width=1.66in]{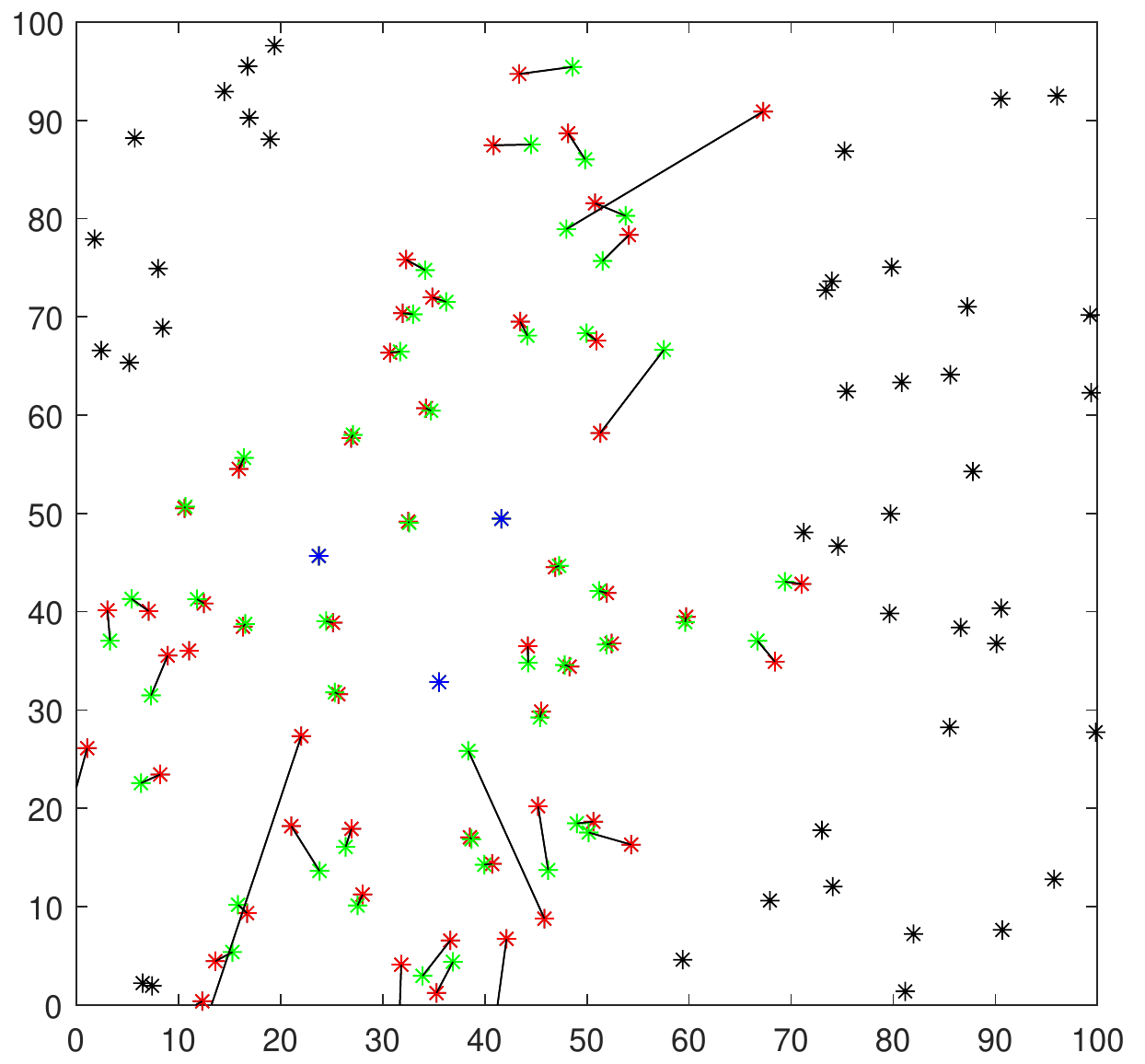}}
  \hspace{0in}
  \subfigure[SELA]{
  \includegraphics[trim=0mm 0mm 0mm 0mm, width=1.66in]{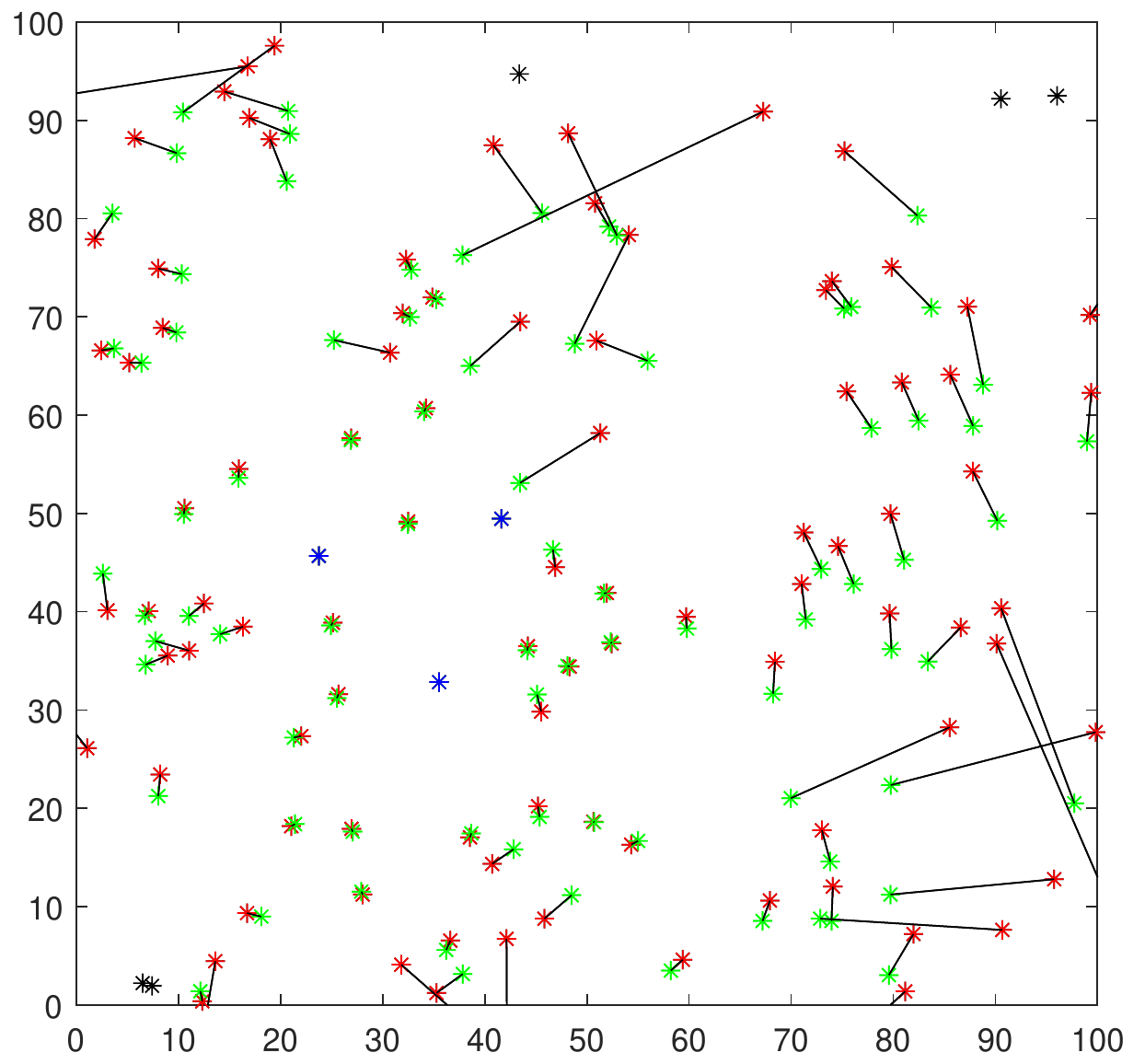}}
  \hspace{0in}
  \subfigure[AFALA]{
  \includegraphics[trim=0mm 0mm 0mm 0mm, width=1.66in]{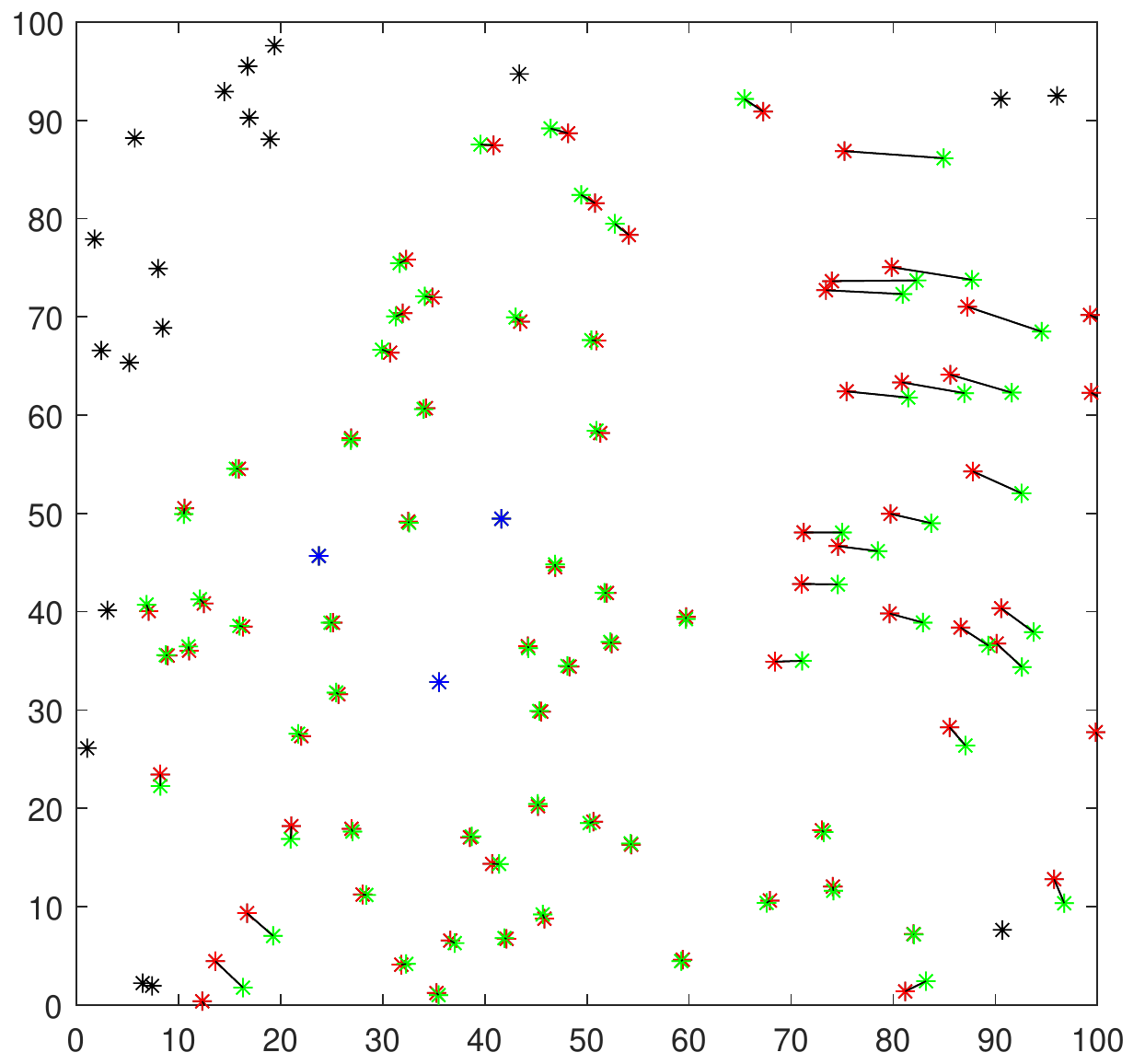}}
  {\caption{Comparison of TLA, SELA and AFALA on a random graph with $r$=20, $\delta$=0.02 and $\varepsilon$=$\delta$, where green and red tic marks denote the true and estimated locations, respectively.}
  \label{fig:4}}
\end{minipage}
\end{figure*}

\begin{figure*}
  \centering
  \subfigure[\%localized nodes]{
  \includegraphics[trim=5mm 0mm -5mm 0mm, width=2.02in]{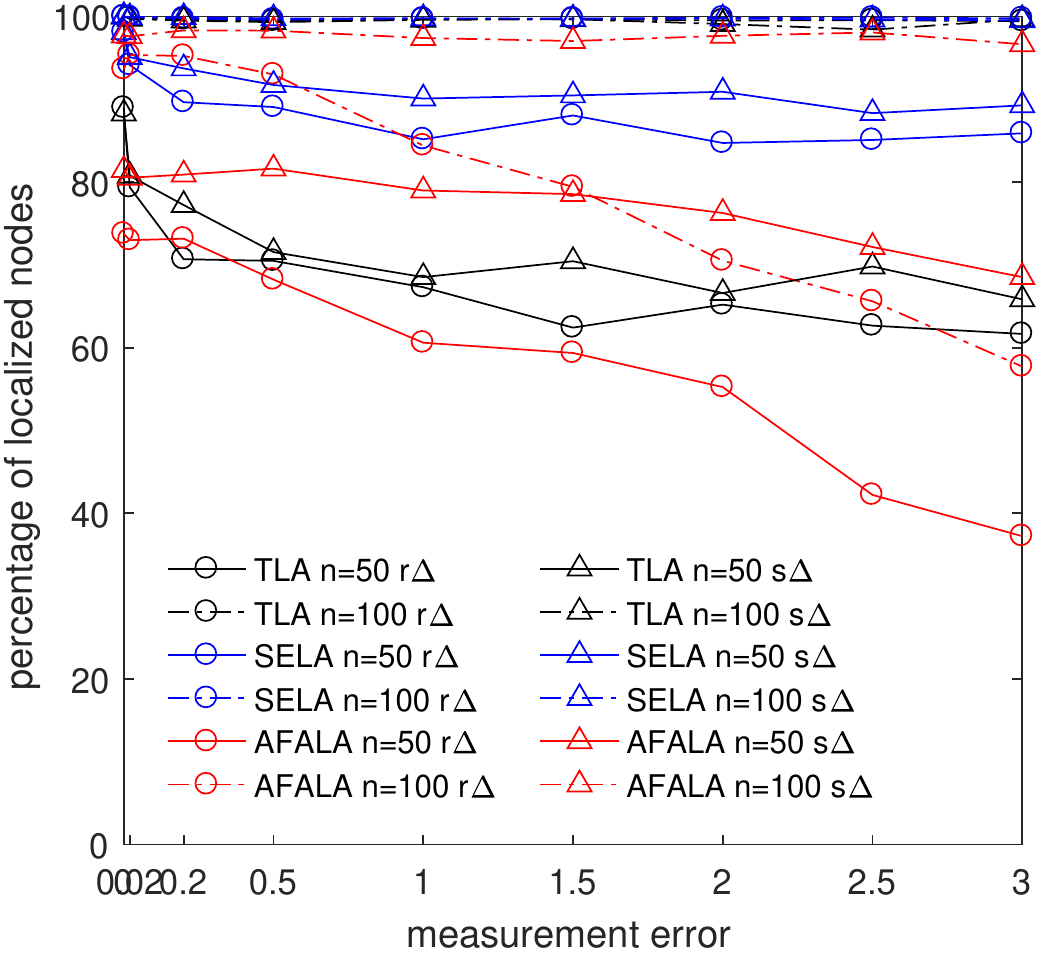}}
  \hspace{0in}
  \subfigure[Average estimation error]{
  \includegraphics[trim=5mm 0mm -5mm 0mm, width=2.02in]{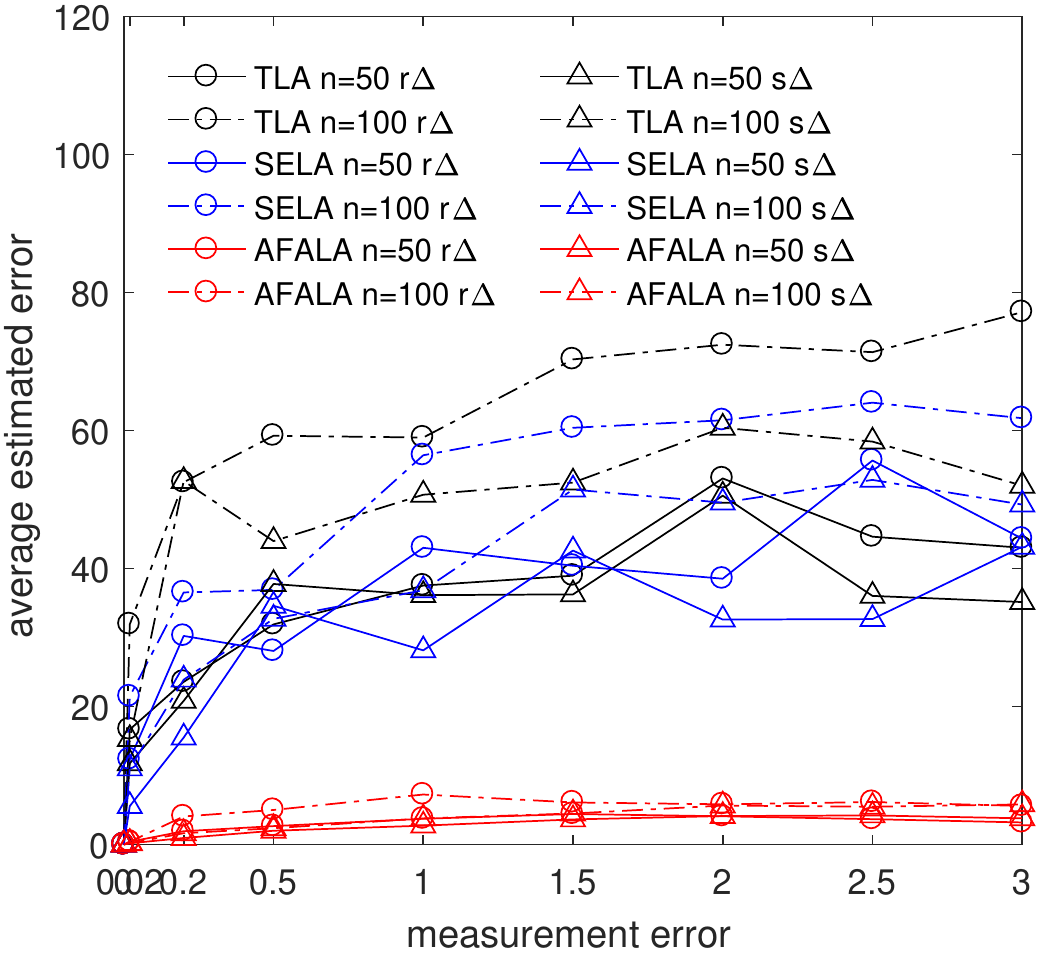}}
  \caption{Impact of initial triangle. $r$ = 30, $d_{max}$ = 0. r$\triangle$ and s$\triangle$ denote the random and special triangle, respectively.}
  \label{fig:5}
\end{figure*}

\begin{figure*}
  \centering
  \subfigure[\%localized nodes]{
  \includegraphics[trim=5mm 0mm -5mm 0mm, width=2.02in]{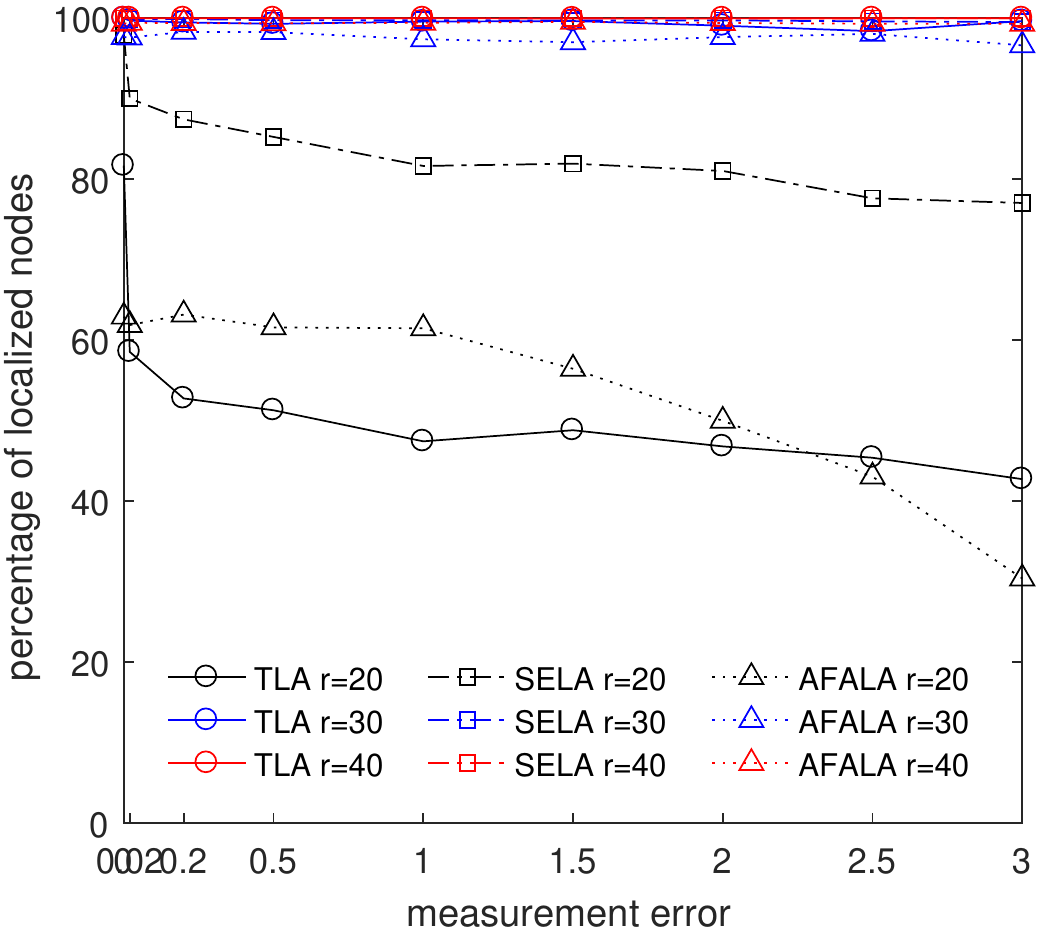}}
  \hspace{0in}
  \subfigure[Average estimation error]{
  \includegraphics[trim=5mm 0mm -5mm 0mm, width=2.02in]{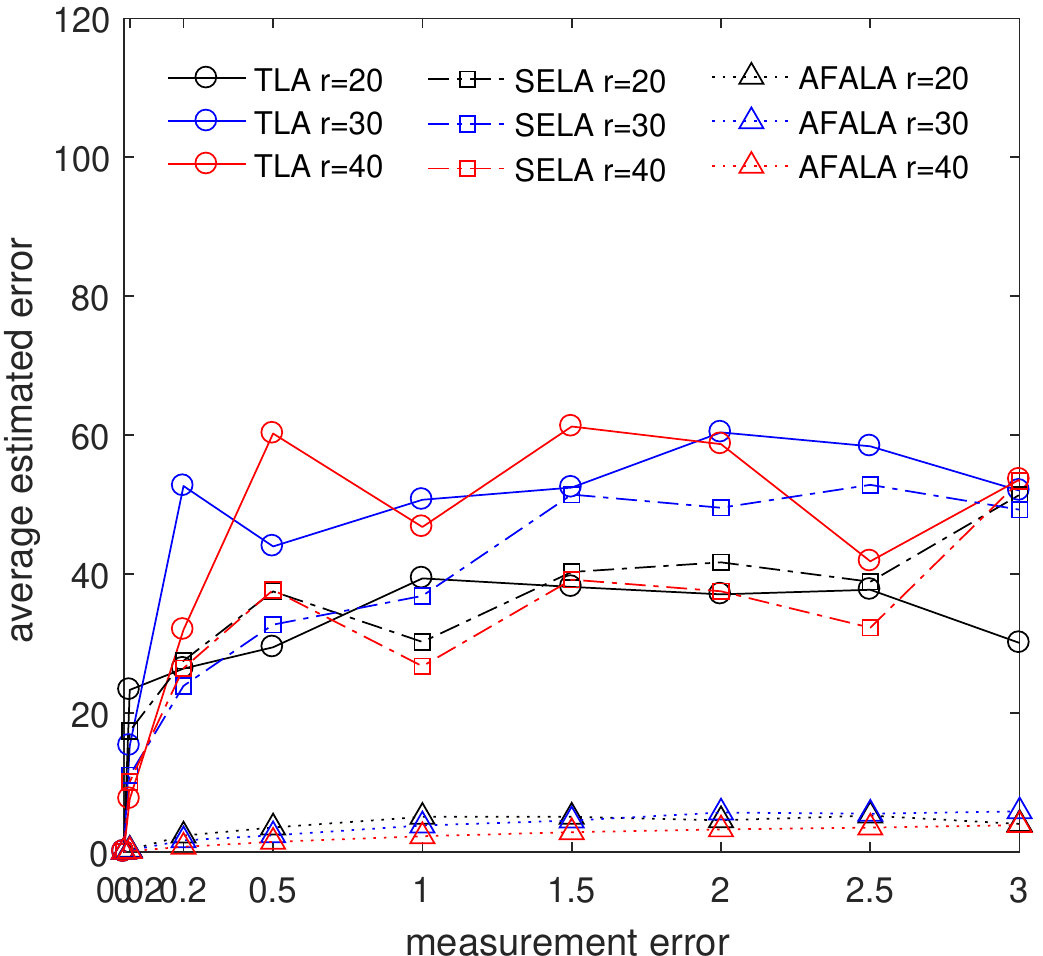}}
  \caption{Impact of measurement error. $n$ = 100, s$\triangle$, $d_{max}$ = 0.}
  \label{fig:6}
\end{figure*}

\begin{figure*}
  \centering
  \subfigure[\%localized nodes]{
  \includegraphics[trim=5mm 0mm -5mm 0mm, width=2.02in]{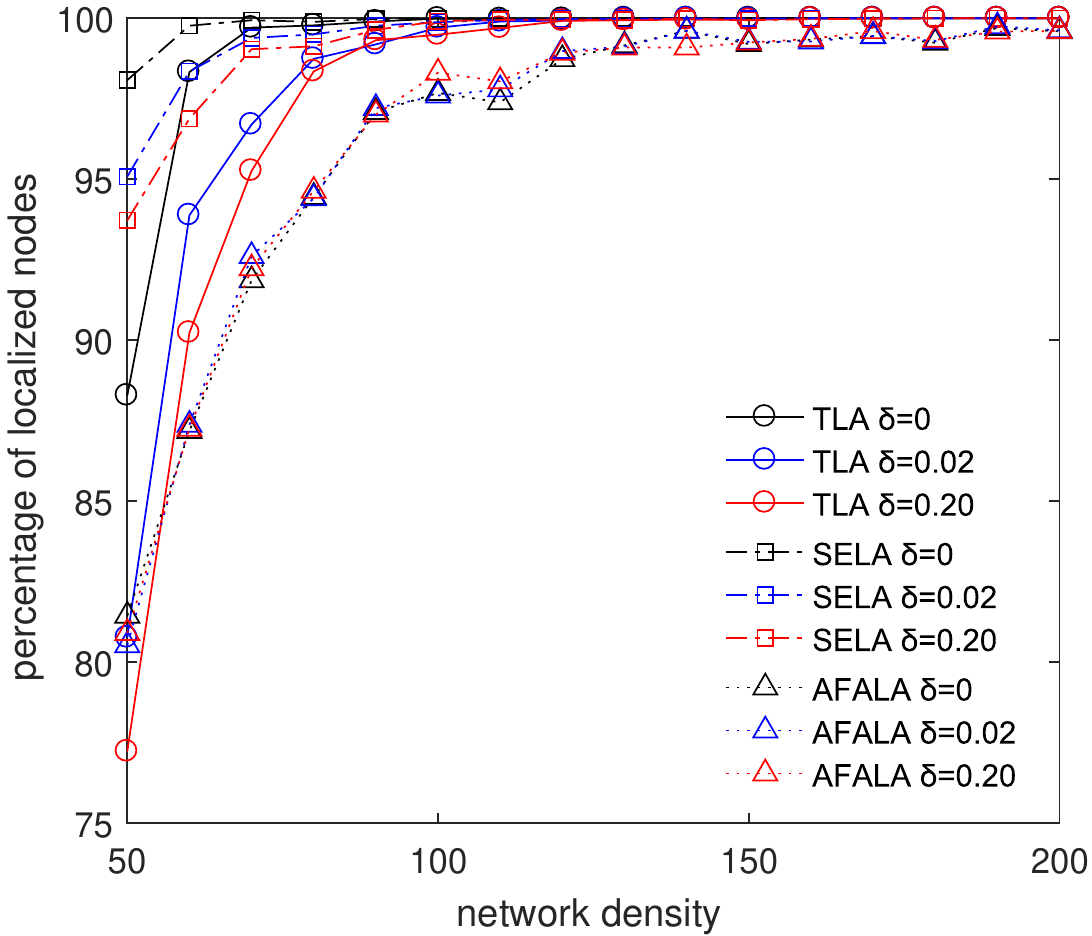}}
  \hspace{0in}
  \subfigure[Average estimation error]{
  \includegraphics[trim=5mm 0mm -5mm 0mm, width=2.02in]{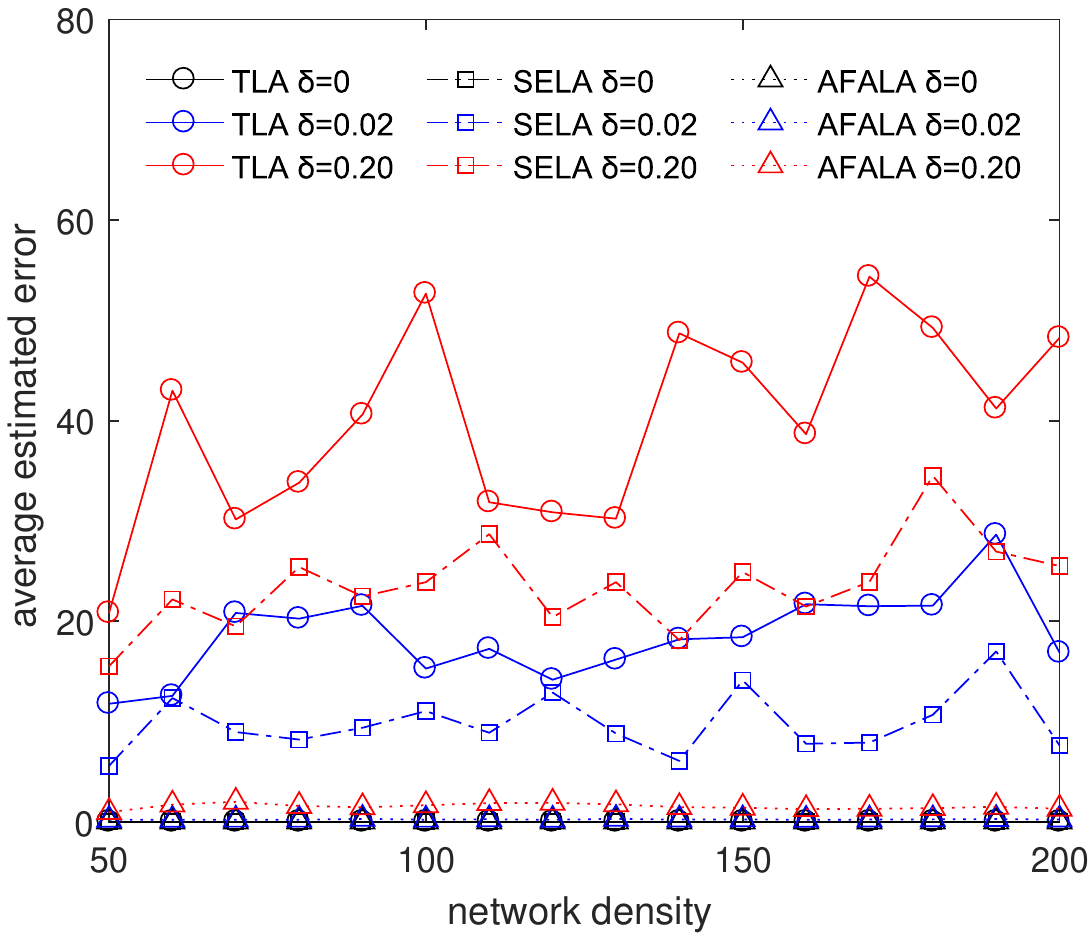}}
  \caption{Impact of network density. s$\triangle$, $r$ = 30, $d_{max}$ = 0.}
  \label{fig:7}
\end{figure*}

\begin{figure*}
  \centering
  \subfigure[\%localized nodes]{
  \includegraphics[trim=5mm 0mm -5mm 0mm, width=2.02in]{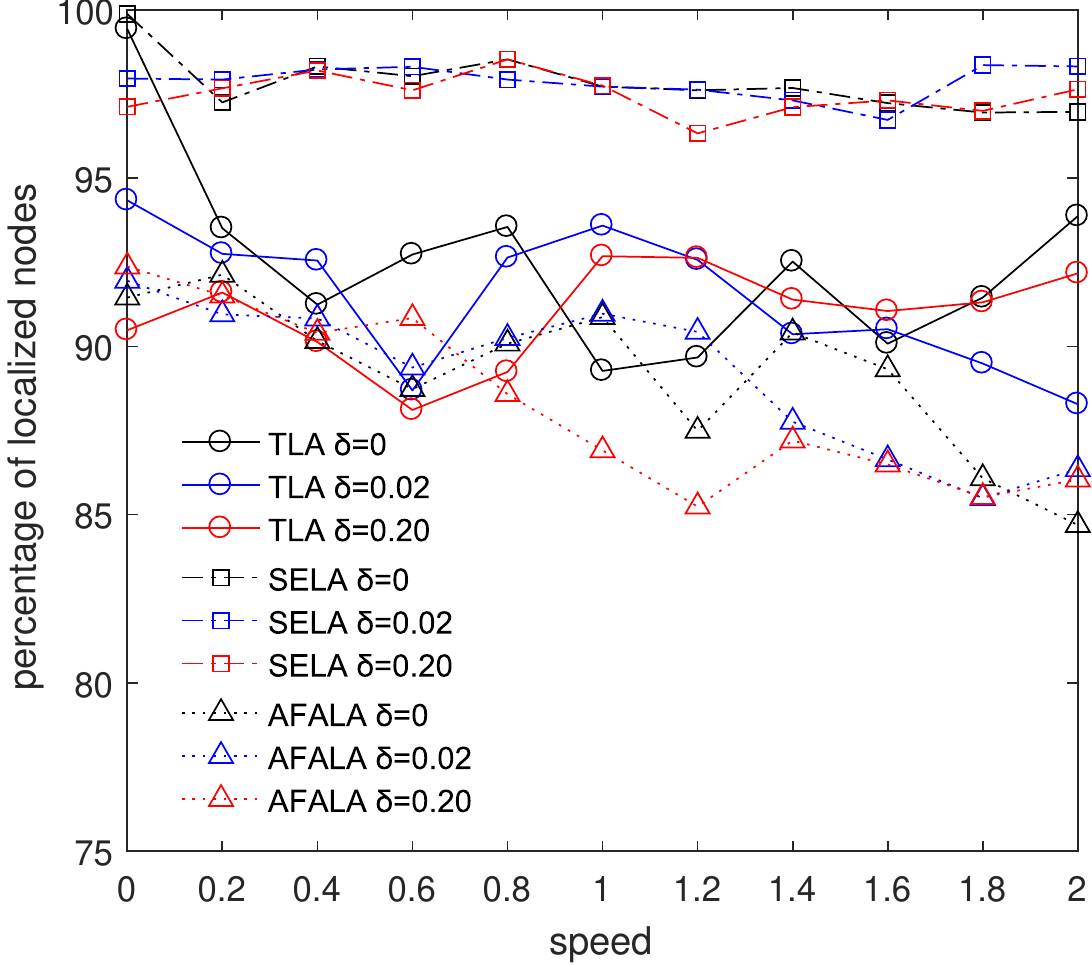}}
  \hspace{0in}
  \subfigure[Average estimation error]{
  \includegraphics[trim=5mm 0mm -5mm 0mm, width=2.02in]{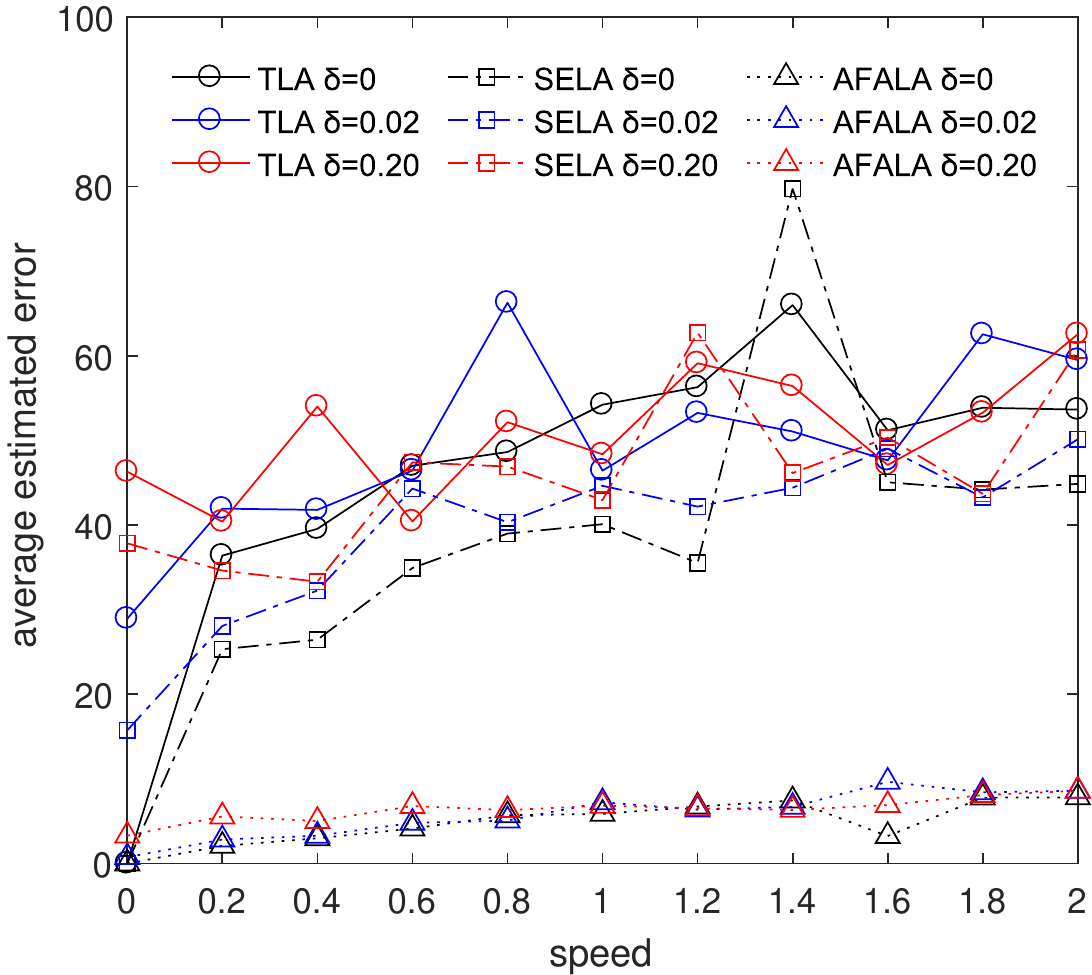}}
  \caption{Impact of node motion. $n$ = 150, s$\triangle$, $r$ = 20.}
  \label{fig:8}
\end{figure*}

\ws{\subsection{Localization Considering Velocities and Directions}}
\ws{Considering the motion velocities and directions of micro-UAVs, we further introduce our localization criteria avoiding FA in trilateration and bilateration. For a connected network, we first assume that the location update can be completed in one time unit. Since the velocities and directions vary, they are represented as the motion distances per time unit \cite{Hu04}. The size and sign of the motion distances are used to denote the magnitude and direction of motions. Here, positive value indicates that a node moves in the opposite direction, leading to an increase of distances between nodes, while negative value indicates the motion towards each other. When the distances between neighboring nodes are acquired together in a time unit, the possible locations of each node fall into a larger region due to their motions.}

\ws{Combining the factor of measurement errors mentioned above, the distance $d_{ij}$ between any two neighboring nodes $s_i$ and $s_j$ can be re-constrained as follows:
\begin{equation}
\label{eqn:30}
  d_{ij} \in [\hat{d}_{ij} - (e_{max} + d_{max}), \hat{d}_{ij} + (e_{max} + d_{max})].
\end{equation}}

\noindent
\ws{These two parameters, measurement error and motion, are regarded as the influential factors on distance-based localizations. Therefore, the dynamic localization problem can be formulated with $\varepsilon = e_{max} + d_{max}$. According to theorem \ref{thm:01} and \ref{thm:02}, two lemmas in trilateration and bilateration can be induced.}

\ws{\begin{lemma}
\label{lem:01}
Given a localized triangle $\triangle opq$, where its three localized nodes $s_o$, $s_p$ and $s_q$ are adjacent to each other, and an unknown node $s_i$, such that the edges $e(o,i)$ and $e(p,i)$ exist, but $e(q,i)$ does not exist. Let $\varepsilon = e_{max} + d_{max} > 0$ be a threshold of the distance variations. $s_i$ can be a unique localization, and its sufficient conditions are as follows:
\begin{enumerate}[\indent(1)]
\item $\forall x\in r_{q1}$, $\forall y\in r_{i1}$, $d_{xy} < D_{qi}$,
\item $\exists u\in r_{q1}$, $\exists v\in r_{i2}$, $d_{uv} > D_{qi}$.
\end{enumerate}
\end{lemma}}

\ws{\begin{lemma}
\label{lem:02}
Given a localized triangle $\triangle opq$, where its three localized nodes $s_o$, $s_p$ and $s_q$ are adjacent to each other, and an unknown node $s_i$, such that the edges $e(o,i)$, $e(p,i)$ and $e(q,i)$ exist. Let $\varepsilon = e_{max} + d_{max} > 0$ be a threshold of the distance variations. $s_i$ can be a unique localization, and its sufficient conditions are as follows:
\begin{enumerate}[\indent(1)]
\item $\forall x\in r_{q1}$, $\forall y\in r_{i1}$, $d_{xy} \notin [\hat{d}_{qi} - \varepsilon, \hat{d}_{qi} + \varepsilon]$,
\item $\exists u\in r_{q1}$, $\exists v\in r_{i2}$, $d_{uv} \in [\hat{d}_{qi} - \varepsilon, \hat{d}_{qi} + \varepsilon]$.
\end{enumerate}
Or,
\begin{enumerate}[\indent(1)]
\item $\forall x\in r_{q1}$, $\forall y\in r_{i2}$, $d_{xy} \notin [\hat{d}_{qi} - \varepsilon, \hat{d}_{qi} + \varepsilon]$,
\item $\exists u\in r_{q1}$, $\exists v\in r_{i1}$, $d_{uv} \in [\hat{d}_{qi} - \varepsilon, \hat{d}_{qi} + \varepsilon]$.
\end{enumerate}
\end{lemma}}

\subsection{Localization Algorithm}
Based on the aforementioned criteria, we now present our \emph{Avoiding Flip Ambiguities Localization Algorithm} (AFALA). \ws{The algorithm first selects a random triangle or a special triangle, which is an acute triangle with every edge greater than four fifths of the communication ranges, as the initial triangle and adds its three nodes to a set of localized nodes.} For any unlocalized node, if it is connected to these three localized nodes and the condition of Theorem \ref{thm:01} or Lemma \ref{lem:01} is satisfied, or it is connected to two localized nodes and the condition of Theorem \ref{thm:02} or Lemma \ref{lem:02} is satisfied, it is localized and is moved from the unlocalized nodes set to the localized nodes set. The process is iterated until no new node can join the set of localized nodes. Details are shown in Algorithm~\ref{alg:01}.

\begin{algorithm}[htbp]
  \begin{algorithmic}[1]
    \REQUIRE ~~\\
      A set of unlocalized nodes, $Su$;\\
      A set of localized triangles, $St$;\\
      The threshold value of distance variations, $\varepsilon$;\\
    \ENSURE ~~\\
      A set of localized nodes, $Sl$;\\
    \STATE choose an initial triangle $\triangle opq$ into $St$;
    \STATE add the three nodes of $\triangle opq$ to $Sl$;
    \REPEAT
        \FOR {each node $s_i\in Su$}
            \FOR {each localized triangle $\triangle opq\in St$}
            \IF {$s_i$ connects to two nodes $s_o$ and $s_p$}
                \STATE $r_{q1}$ = $r_{oq} \cap r_{pq} \cap H$;
                \STATE $r_{i1}$ = $r_{oi} \cap r_{pi} \cap H$;
                \STATE $r_{i2}$ = $r_{oi} \cap r_{pi} \cap H'$;
                \STATE $[G_{q1},G_{i1},G_{i2}]$ = the set of grid point coordinate of $r_{q1}$, $r_{i1}$ and $r_{i2}$, respectively;
                \IF {$s_i$ connects to the node $s_q$}
                    \IF {(($\forall x\in G_{q1}$, $\forall y\in G_{i1}$, $d_{xy}< \hat{d}_{qi} - \varepsilon$ $\wedge$ $d_{xy}> \hat{d}_{qi} + \varepsilon$)
                    AND ($\exists u\in G_{q1}$, $\exists v\in G_{i2}$, $\hat{d}_{qi} - \varepsilon < d_{uv} < \hat{d}_{qi} + \varepsilon$))\\
                    OR  (($\exists x\in G_{q1}$, $\exists y\in G_{i1}$, $\hat{d}_{qi} - \varepsilon < d_{xy} < \hat{d}_{qi} + \varepsilon$)
                    AND ($\forall u\in G_{q1}$, $\forall v\in G_{i2}$, $d_{uv}< \hat{d}_{qi} - \varepsilon$ $\wedge$ $d_{uv}> \hat{d}_{qi} + \varepsilon$))}
                        \STATE calculate the estimated location $\hat{p}_i$;
                        \STATE move $s_i$ from $Su$ to $Sl$;
                        \STATE update the distances of the localized subgraph
                        \STATE add $\triangle opi$ to $St$;
                    \ENDIF
                \ELSE
                    \STATE $D_{qi}$ = $min(max(D_q^1) - \varepsilon, max(D_i^1) - \varepsilon)$;
                    \IF {($\forall x\in G_{q1}$, $\forall y\in G_{i1}$, $d_{xy}< D_{qi}$) AND $(\exists u\in G_{q1}$, $\exists v\in G_{i2}$, $d_{uv}> D_{qi}$)}
                        \STATE calculate the estimated location $\hat{p}_i$;
                        \STATE move $s_i$ from $Su$ to $Sl$;
                        \STATE update the distances of the localized subgraph
                        \STATE add $\triangle opi$ to $St$;
                    \ENDIF
                \ENDIF
            \ENDIF
            \ENDFOR
        \ENDFOR
    \UNTIL {the element number of $Sl$ no longer changes}
    \RETURN $Sl$;
  \end{algorithmic}
  \caption{the AFALA algorithm based on both localization criteria.}
  \label{alg:01}
\end{algorithm}

\section{Performance Validation}\label{sec:PerformanceValidation}
In this section, the localization performance of the proposed algorithm is compared with TLA \cite{Eren04} and SELA \cite{Olivia15} in terms of the percentage of localizable nodes and the average estimation error, using simulations implemented in MATLAB 7.14.0.739 (R2012a). Simulations are conducted in a square unit area of 100 $units$ by 100 $units$, where micro-UAVs are uniformly distributed. 100 random instances of micro-UAV network are administered in each group trial, and the average result is taken to ensure more accurate result. \ws{The experimental parameter settings are as follows:}
\ws{\begin{itemize}
\item The number $n$ of micro-UAVs deployed in networks ranges from 50 to 200;
\item A random or special triangle is randomly chosen as the initial triangle;
\item The communication radius $r$ is set to be 20, 30 and 40 $units$; 
\item The measurement errors $e_{ij} \sim N(0, \delta^2)$, such that $\hat{d}_{ij} = |d_{ij} + e_{ij}|$, where $\delta$ is set to be 0, 0.02, 0.2, 0.5, 1, 1.5, 2, 2.5 and 3 $units$;
\item The motion speed of each micro-UAV is randomly chosen from [$-d_{max}$, $d_{max}$] units per time unit, where the maximum motion speed $d_{max}$ is set to be 0:0.2:2 $units$ per time unit;
\item The motion factor is considered or not by the threshold value $\varepsilon$ which is set to be $\delta$ or $\delta + d_{max}$. In order to keep the errors within the maximum allowable bounds, if $|d_{ij} - \hat{d}_{ij}| > \varepsilon$, then $|d_{ij} - \hat{d}_{ij}| = \varepsilon$.
\end{itemize}}

Fig. \ref{fig:4} presents a representative localization scenario with $r$ = 20 and $\delta$ = 0.2. Fig. \ref{fig:4} (a) shows the original graph of the network. Fig. ~\ref{fig:4} (b), (c) and (d) present the localization results of the TLA, SELA and AFALA algorithms in the same scenario. The three blue nodes are randomly chosen as the nodes of the initial triangle. The green and red nodes are the true and estimated locations, respectively, which are connected with black lines, while the remaining black parts are the unlocated nodes. Obviously, the SELA and AFALA algorithms locate more nodes than TLA because of the utilization of bilateration. However, many FAs occur in the TLA and SELA algorithms, while no FAs occur in AFALA. That is because our algorithm, AFALA, fully consider the effect of measurement error on FA. Therefore, our algorithm locates more accurately than TLA and SELA.

We further demonstrate the localization performance in terms of percentage of localized nodes and average estimation error by comparing AFALA with TLA and SELA from aspects such as initial triangle, measurement error, network density and node motion.

\ws{\subsection{Initial Triangle: Random \& Special Triangle}}
\ws{As the starting point of algorithms, the choice of initial triangles determines whether algorithms can localize their first new node. If the node fails to be localized, that means the localization process was over before initiating. Fig. \ref{fig:5} shows the effect of random and special triangle on the localization performance. As is shown in Fig. \ref{fig:5} (a), for every network density, these three algorithms (especially AFALA) with special triangle distinctly localizes more nodes than with random triangle under the same measurement errors. Furthermore, with the increase of measurement errors, the number of nodes localized by AFALA algorithm with random triangle rapidly drops. In contrast, the special triangle slowly decreases the number of localized nodes. Especially, when network density is equal to 100, the number of localized nodes with special triangle almost remains stable, even reaching to about 100\%. Theoretically, the special triangle more easily satisfies the localization criteria of AFALA, thus improving the number of localized nodes.}

\ws{Fig. \ref{fig:5} (b) shows the effect of random and special triangle on the average estimation error. Obviously, the initial triangle tends to have less impact on AFALA than both TLA and SELA. Therefore, the choice of the initial triangle significantly affects the localization result of these methods, and special triangle achieves better localization performance than random triangle in terms of the localization ratio.}

\subsection{Measurement Error}
We demonstrate the impact of measurement errors on the localization performance in terms of percentage of localized nodes and average estimation error in Fig. \ref{fig:6}. Fig. \ref{fig:6} (a) shows the percentage of localized nodes against the measurement error. For $r$ = 20, the AFALA and SELA algorithms locate more nodes than TLA when the measurement error is small, especially SELA. However, as the measurement error increases, the percentage of localized nodes decreases markedly in AFALA. That is mainly because the strict localization conditions of AFALA are difficult to satisfy with small communication radius and large measurement errors. For the large communication radius $r$ = 30 and 40, the performance gap among the three algorithms decreases gradually. Especially for $r$ = 40, they are able to localize almost all the nodes. That is because the large communication radius increases the possibility of using trilateration and reduces the key differences among these algorithms. For the case of AFALA, the localization conditions are also easily satisfied, leading to good outcomes.

Fig. \ref{fig:6} (b) shows the average estimation error against the measurement error where AFALA performs significantly better than TLA and SELA. For every value of $r$, it is apparent that the average estimation error is much less in AFALA than in TLA and SELA. Thus, the proposed algorithm exhibits excellent localization performance in terms of the average estimation error. This confirms the efficacy of AFALA as it addresses the critical issue of measurement errors, a key source of FA, to fulfil the criteria for accurate localization.

\ws{\subsection{Network Density}}
\ws{Fig. \ref{fig:7} illustrates the impact of network density on the localization performance in different localization algorithms. Fig. \ref{fig:7} (a) shows the impact of network density on the localization ratio. For TLA, SELA and AFALA, the number of localized nodes significantly increases with the increase of network density. SELA achieves the fastest growth due to its most flexible localization conditions. On the contrary, AFALA obtains the lowest growth because of its most rigorous conditions. However, the gap among them is gradually reduced with the increase of network density. When network density is larger than 120, AFALA localizes almost all the nodes, reaching nearly the same localization ratio of TLA and SELA.}

\ws{Fig. \ref{fig:7} (b) shows the impact of network density on the average estimation error. Although these three algorithms are comparable in the percentage of localized nodes, there is an obvious difference among them in the average estimation error. As is shown in Fig. \ref{fig:7} (b), the performance curves fluctuate with a fixed estimation error. For any one of these network densities, AFALA outperforms both TLA and SELA by large margins when the measurement error occurs.}

\ws{\subsection{Node Motion}}
\ws{We demonstrate the effect of node motion on the localization performance of TLA, SELA and AFALA in Fig. \ref{fig:8}. As is shown in Fig. \ref{fig:8} (a), compared with TLA, our AFALA locates similar number of nodes, and the motion speed tends to slightly decrease the number of located nodes. That is because our localization criteria depend on the measurement distances, the measurement error and the motion speed to estimate the size of intersection regions for the avoidance of flip ambiguities, which makes it sensitive to these factors.}

\ws{Fig. \ref{fig:8} (b) shows the impact of motion speed on the average estimation error. The average measurement error increases as the maximum motion speed increases. Our localization criteria avoiding flip ambiguities improve the localization accuracy, thus making AFALA to achieve much less average measurement error than TLA and SELA. Therefore, AFALA outperforms TLA and SELA in terms of localization accuracy without obvious loss of localization ratio.}

\section{Conclusion}\label{sec:Conclusion}
In this paper, we have proposed a localization algorithm for swarms micro-UAVs that aims to reduce the occurrence of flip ambiguities (FA). This is critical for collaborative flight of micro-UAVs to prevent collisions that can arise from localization errors caused by FA. For both bilateration and trilateration, under conditions of bounded errors, we analyzed the FA phenomenons using the characteristics of intersecting regions and derived the localization criteria for avoiding FA theoretically. Using these criteria, we developed a corresponding localization algorithm, which we call the \emph{Avoiding Flip Ambiguities Localization Algorithm} (AFALA). \ws{Using simulations implemented in MATLAB, we compared the performance of AFALA against other well-known localization methods, viz. TLA and SELA, to demonstrate its efficiency from four aspects of initial triangle, measurement error, network density and node motion.}

\section*{Acknowledgment}

This work is supported by the National Natural Science Foundation of China (61572231), by the Shandong Provincial Key Research \& Development Project (2017GGX10141), and by Natural Science Foundation of Shandong Province of China (ZR2017BF016).

\bibliographystyle{IEEEtran}
\bibliography{mywsnlib}

\begin{thebibliography}{10}
\providecommand{\url}[1]{#1}
\csname url@samestyle\endcsname
\providecommand{\newblock}{\relax}
\providecommand{\bibinfo}[2]{#2}
\providecommand{\BIBentrySTDinterwordspacing}{\spaceskip=0pt\relax}
\providecommand{\BIBentryALTinterwordstretchfactor}{4}
\providecommand{\BIBentryALTinterwordspacing}{\spaceskip=\fontdimen2\font plus
\BIBentryALTinterwordstretchfactor\fontdimen3\font minus
  \fontdimen4\font\relax}
\providecommand{\BIBforeignlanguage}[2]{{%
\expandafter\ifx\csname l@#1\endcsname\relax
\typeout{** WARNING: IEEEtran.bst: No hyphenation pattern has been}%
\typeout{** loaded for the language `#1'. Using the pattern for}%
\typeout{** the default language instead.}%
\else
\language=\csname l@#1\endcsname
\fi
#2}}
\providecommand{\BIBdecl}{\relax}
\BIBdecl

\bibitem{Youssef05}
A.~Youssef, A.~Agrawala, and M.~Younis, ``Accurate anchor-free node
  localization in wireless sensor networks,'' in \emph{Proceedings of the 24th
  IEEE International Performance, Computing, and Communications Conference},
  ser. PCCC 2005, Phoenix, AZ, USA, USA, 2005, pp. 465--470.

\bibitem{Cheng04}
X.~Cheng, A.~Thaeler, G.~Xue, and D.~Chen, ``Tps: A time-based positioning
  scheme for outdoor wireless sensor networks,'' in \emph{Proceedings of the
  23th IEEE INFOCOM}, Hong Kong, China, 2004, pp. 2685--2696.

\bibitem{Patwari03}
N.~Patwari, A.~O. Hero, M.~Perkins, N.~S. Correal, and R.~J. O'Dea, ``Relative
  location estimation in wireless sensor networks,'' \emph{IEEE Transactions on
  Signal Processing}, vol.~51, no.~8, pp. 2137--2148, 2003.

\bibitem{Potdar09}
V.~Potdar, A.~Sharif, and E.~Chang, ``Wireless sensor networks: A survey,'' in
  \emph{Proceedings of the 2009 International Conference on Advanced
  Information Networking and Applications Workshops}, ser. WAINA '09,
  Washington, DC, USA, 2009, pp. 636--641.

\bibitem{Carter81}
G.~C. GARTER, ``Time delay estimation for passive sonar signal processing,''
  \emph{IEEE Transactions on Acoustics, Speech, and Signal Processing},
  vol.~29, no.~3, pp. 463--470, 1981.

\bibitem{Rappaport96}
T.~S. Rappaport, J.~H. Reed, and B.~D. Woerner, ``Position location using
  wireless communications on highways of the future,'' \emph{IEEE
  Communications Magazine}, vol.~34, no.~10, pp. 33--41, 1996.

\bibitem{Bernhardt87}
R.~Bernhardt, ``Macroscopic diversity in frequency reuse radio systems,''
  \emph{IEEE Journal on Selected Areas in Communications}, vol.~5, no.~5, pp.
  862--870, 1987.

\bibitem{Lee02}
J.-Y. Lee and R.~A. Scholtz, ``Ranging in a dense multipath environment using
  an uwb radio link,'' \emph{IEEE Journal on Selected Areas in Communications},
  vol.~20, no.~9, pp. 1677--1683, 2006.

\bibitem{Gezici05}
S.~Gezici, Z.~Tian, G.~B. Giannakis, H.~Kobayashi, A.~F. Molisch, H.~V. Poor,
  and Z.~Sahinoglu, ``Localization via ultra-wideband radios: a look at
  positioning aspects for future sensor networks,'' \emph{IEEE Signal
  Processing Magazine}, vol.~22, no.~4, pp. 70--84, 2005.

\bibitem{Severi09}
S.~Severi, G.~Abreu, G.~Destino, and D.~Dardari, ``Understanding and solving
  flip-ambiguity in network localization via semidefinite programming,'' in
  \emph{Proceedings of the 28th IEEE Conference on Global Telecommunications},
  ser. GLOBECOM'09, Piscataway, NJ, USA, 2009, pp. 3910--3915.

\bibitem{Wang11}
X.~Wang, Z.~Yang, J.~Luo, and C.~Shen, ``Beyond rigidity: Obtain localisability
  with noisy ranging measurement,'' \emph{Int. J. Ad Hoc Ubiquitous Comput.},
  vol.~8, no. 1/2, pp. 114--124, 2011.

\bibitem{Liu14}
W.~Liu, E.~Dong, Y.~Song, and D.~Zhang, ``An improved flip ambiguity detection
  algorithm in wireless sensor networks node localization,'' in
  \emph{Proceedings of the 21st International Conference on Telecommunications
  (ICT)}, Lisbon, Portugal, 2014, pp. 206--212.

\bibitem{Yang16}
S.~Yang, D.~Enqing, L.~Wei, and P.~Xue, ``An iterative method of processing
  node flip ambiguity in wireless sensor networks node localization,'' in
  \emph{Proceedings of the 2016 International Conference on Information
  Networking (ICOIN)}, Kota Kinabalu, Malaysia, 2016, pp. 92--97.

\bibitem{Zhang12_2}
Y.~Zhang, Y.~Chen, and Y.~Liu, ``Towards unique and anchor-free localization
  for wireless sensor networks,'' \emph{Wireless Personal Communications},
  vol.~63, no.~1, pp. 261--278, 2012.

\bibitem{Niculescu04}
D.~Niculescu and B.~Nath, ``Error characteristics of ad hoc positioning systems
  (aps),'' in \emph{Proceedings of the 5th ACM International Symposium on
  Mobile Ad Hoc Networking and Computing}, Tokyo, Japan, 2004, pp. 20--30.

\bibitem{Moore04}
D.~Moore, J.~Leonard, D.~Rus, and S.~Teller, ``Robust distributed network
  localization with noisy range measurements,'' in \emph{Proceedings of the 2Nd
  International Conference on Embedded Networked Sensor Systems}, ser. SenSys
  '04, New York, NY, USA, 2004, pp. 50--61.

\bibitem{Kannan10}
A.~A. Kannan, B.~Fidan, and G.~Mao, ``Analysis of flip ambiguities for robust
  sensor network localization,'' \emph{IEEE Transactions on Vehicular
  Technology}, vol.~59, no.~4, pp. 2057 -- 2070, 2010.

\bibitem{Zhang12_1}
Y.~Zhang, S.~Liu, X.~Zhao, and Z.~Jia, ``Theoretic analysis of unique
  localization for wireless sensor networks,'' \emph{Ad Hoc Netw.}, vol.~10,
  no.~3, pp. 623--634, 2012.

\bibitem{Aspnes06}
J.~Aspnes, T.~Eren, D.~K. Goldenberg, A.~S. Morse, W.~Whiteley, Y.~R. Yang,
  B.~D.~O. Anderson, and P.~N. Belhumeur, ``A theory of network localization,''
  \emph{IEEE Transactions on Mobile Computing}, vol.~5, no.~12, pp. 1663--1678,
  2006.

\bibitem{Anderson10}
B.~D.~O. Anderson, I.~Shames, G.~Mao, and B.~Fidan, ``Formal theory of noisy
  sensor network localization,'' \emph{SIAM J. Discret. Math.}, vol.~24, no.~2,
  pp. 684--698, 2010.

\bibitem{Eren04}
T.~Eren, O.~K. Goldenberg, W.~Whiteley, Y.~R. Yang, A.~S. Morse, and B.~D. O.,
  ``Rigidity, computation, and randomization in network localization,'' in
  \emph{Proceedings of the IEEE INFOCOM 2004}, Hong Kong, China, 2004, pp. 2673
  -- 2684.

\bibitem{Akcan13}
H.~Akcan and C.~Evrendilek, ``Reduce the number of flips in trilateration with
  noisy range measurements,'' in \emph{Proceedings of the 12th International
  ACM Workshop on Data Engineering for Wireless and Mobile Acess}, ser. MobiDE
  '13, New York, NY, USA, 2013, pp. 20--27.

\bibitem{Kannan08}
A.~A. Kannan, B.~Fidan, and G.~Mao, ``Robust distributed sensor network
  localization based on analysis of flip ambiguities,'' in \emph{Proceedings of
  the 2008 IEEE Global Telecommunications Conference}, New Orleans, LO, USA,
  2008, pp. 1--6.

\bibitem{Fang06}
J.~Fang, M.~Cao, A.~S. Morse, and B.~D.~O. Anderson, ``Localization of sensor
  networks using sweeps,'' in \emph{Proceedings of the 45th IEEE Conference on
  Decision and Control}, San Diego, CA, USA, 2006, pp. 4645--4650.

\bibitem{Goldenberg06}
D.~K. Goldenberg, P.~Bihler, M.~Cao, J.~Fang, B.~D.~O. Anderson, A.~S. Morse,
  and Y.~R. Yang, ``Localization in sparse networks using sweeps,'' in
  \emph{Proceedings of the 12th Annual International Conference on Mobile
  Computing and Networking}, ser. MobiCom '06, New York, NY, USA, 2006, pp.
  110--121.

\bibitem{Yang09}
Z.~Yang, Y.~Liu, and X.~Y. Li, ``Beyond trilateration: on the localizability of
  wireless ad-hoc networks,'' in \emph{Proceedings of the IEEE INFOCOM 2009},
  Rio de Janeiro, Brazil, 2009, pp. 2392--2400.

\bibitem{Olivia15}
G.~Oliva, S.~Panzieri, F.~Pascucci, and R.~Setola, ``Sensor networks
  localization: extending trilateration via shadow edges,'' \emph{IEEE
  Transactions on Automatic Control}, vol.~60, no.~10, pp. 2752--2755, 2015.

\bibitem{Kaewprapha11}
P.~Kaewprapha, J.~Li, and N.~Puttarak, ``Network localization on unit disk
  graphs,'' in \emph{Proceedings of the 2011 IEEE Global Telecommunications
  Conference - GLOBECOM 2011}, Kathmandu, Nepal, 2011, pp. 1--5.

\bibitem{Kuhn04}
F.~Kuhn, T.~Moscibroda, and R.~Wattenhofer, ``Unit disk graph approximation,''
  in \emph{Proceedings of the 2004 Joint Workshop on Foundations of Mobile
  Computing}, ser. DIALM-POMC '04, New York, NY, USA, 2004, pp. 17--23.

\bibitem{Iwashige15}
J.~Iwashige, L.~Barolli, S.~Kameyama, and M.~Iwaida, ``Diffracted fields in
  buildings and wedges: A comparison study,'' in \emph{Proceedings of the 2013
  Eighth International Conference on Broadband and Wireless Computing,
  Communication and Applications}, ser. BWCCA '13, Compiegne, France, 2013, pp.
  408--413.

\bibitem{Shi17}
X.~Shi, G.~Mao, B.~D.~O. Anderson, Z.~Yang, and J.~Chen, ``Robust localization
  using range measurements with unknown and bounded errors,'' \emph{IEEE
  Transactions on Wireless Communications}, vol.~16, no.~6, pp. 4065--4078,
  2017.

\bibitem{Hu04}
L.~Hu and D.~Evans, ``Localization for mobile sensor networks,'' in
  \emph{Proceedings of the 10th Annual International Conference on Mobile
  Computing and Networking}, ser. MobiCom '04, New York, NY, USA, 2004, pp.
  45--57.

\bibitem{Evrendilek11}
C.~Evrendilek and H.~Akcan, ``On the complexity of trilateration with noisy
  range measurements,'' \emph{IEEE Communications Letters}, vol.~15, no.~10,
  pp. 1097 -- 1099, 2011.

\bibitem{Savarese02}
C.~Savarese, J.~M. Rabaey, and K.~Langendoen, ``Robust positioning algorithms
  for distributed ad-hoc wireless sensor networks,'' in \emph{Proceedings of
  the General Track of the Annual Conference on USENIX Annual Technical
  Conference}, ser. ATEC '02, Berkeley, CA, USA, 2002, pp. 317--327.

\bibitem{Sittile08}
F.~Sottile and M.~A. Spirito, ``Robust localization for wireless sensor
  networks,'' in \emph{Proceedings of the 2008 5th Annual IEEE Communications
  Society Conference on Sensor, Mesh and Ad Hoc Communications and Networks},
  San Francisco, CA, USA, 2008, pp. 46--54.

\bibitem{Hendrickson92}
B.~Hendrickson, ``Conditions for unique graph realizations,'' \emph{SIAM J.
  Comput}, vol.~21, no.~1, pp. 65--84, 1992.

\bibitem{Zhang10}
S.~Zhang, J.~Cao, C.~Li-Jun, and D.~Chen, ``Accurate and energy-efficient
  range-free localization for mobile sensor networks,'' \emph{IEEE Transactions
  on Mobile Computing}, vol.~9, no.~6, pp. 897--910, 2010.

\bibitem{Baggio08}
B.~A and L.~K, ``Monte-carlo localization for mobile wireless sensor
  networks,'' in \emph{Proceedings of the International Conference on Mobile
  Ad-Hoc and Sensor Networks}, Berlin, Heidelberg, 2008, pp. 317--328.

\bibitem{Keung10}
G.~Y. Keung, B.~Li, and Q.~Zhang, ``Message delivery capacity in
  delay-constrained mobile sensor networks: bounds and realization,''
  \emph{IEEE Transactions on Wireless Communications}, vol.~10, no.~5, pp.
  1552--1559, 2011.

\end{thebibliography}

\begin{IEEEbiography}[{\includegraphics[width=1in,height=1.125in,clip,keepaspectratio]{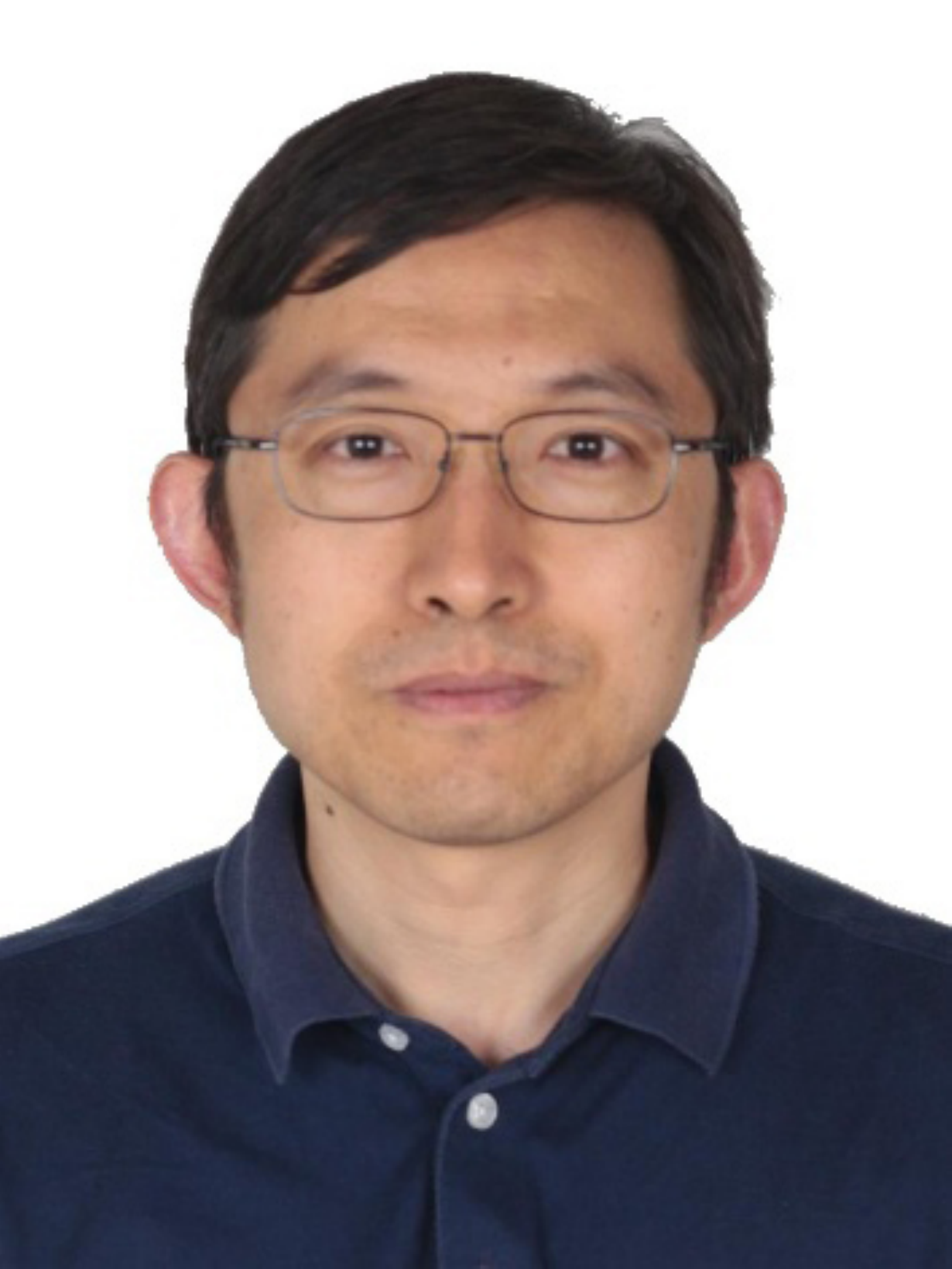}}]{Qingbei Guo}
received the M.S. degree from the School of Computer Science and Technology, Shandong University, Jinan, China, in 2006. He is a member of the Shandong Provincial Key Laboratory of Network based Intelligent Computing and the lecturer in the School of Information Science and Engineering, University of Jinan. He is now a Ph.D. student at the Centre for Vision, Speech and Signal Processing, Jiangnan University, Wuxi, China. His current research interests include wireless sensor networks, deep learning/machine learning, computer vision and neuron networks.
\end{IEEEbiography}

\begin{IEEEbiography}[{\includegraphics[width=1in,height=1.125in,clip,keepaspectratio]{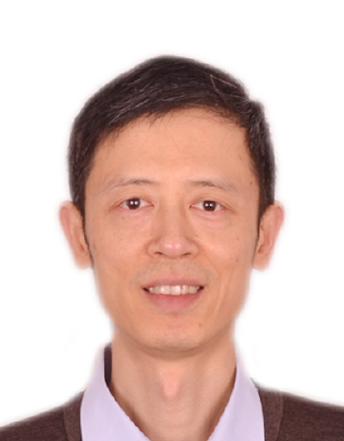}}]{Yuan Zhang}
received his M.S. degree in Communication Systems and Ph.D. degree in Control Theory \& Engineering (Biomedical Engineering) both from Shandong University, China, in 2003 and 2012 respectively. He is currently an Associate Professor at University of Jinan, China. Dr. Zhang was a visiting professor at Computer Science Department, Georgia State University, USA, in 2014. As the first author or corresponding author he has published more than 50 peer reviewed papers in international journals and conference proceedings, 1 book chapters, and 6 patents in theareas of Smart Health and Biomedical Big Data Analytics. He has served as Leading Guest Editor for six special issues of IEEE, Elsevier, Springer and InderScience publications, including IEEE Internet of Things Journal special issue on Wearable Sensor Based Big Data Analysis for Smart Health and IEEE Journal of Biomedical and Health Informatics (JBHI)special issue on Pervasive Sensing and Machine Learning for Mental Health. He has served on the technical program committee for numerous international conferences. He is an associate editor for IEEE Access. Dr.Zhang’s research interests are Wearable Sensing for Smart Health, Machine Learning for Auxiliary Diagnosis, and Biomedical Big Data Analytics. His research has been extensively supported by the Natural Science Foundation of China,China Postdoctoral Science Foundation, and Natural Science Foundation of Shandong Province with total grant funding over 1.4 million RMB. Dr. Zhang is a Senior Member of both IEEE and ACM. For more information, please refer to http://uslab.ujn.edu.cn/index.html
\end{IEEEbiography}


\begin{IEEEbiography}[{\includegraphics[width=1in,height=1.125in,clip,keepaspectratio]{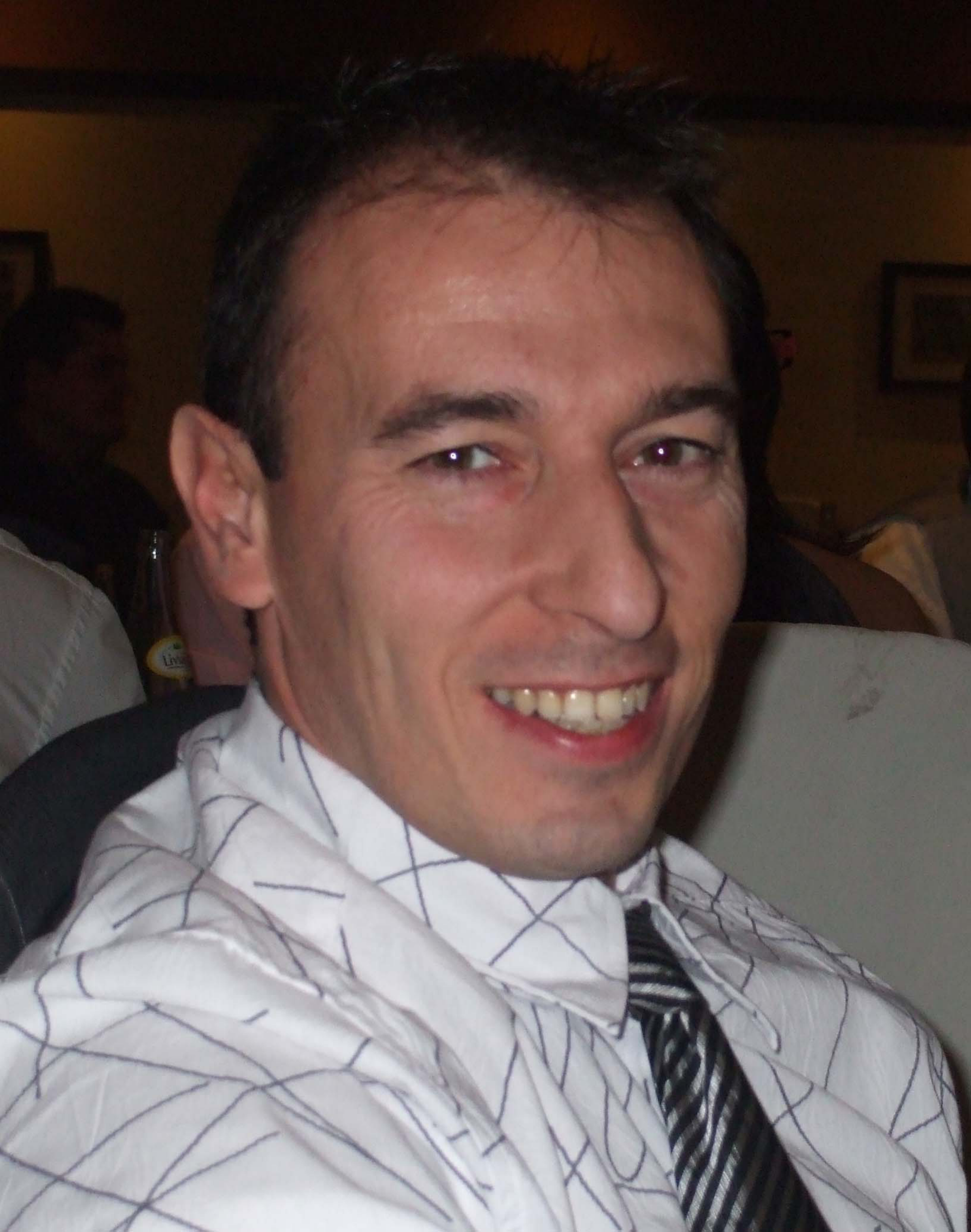}}]{Jaime Lloret}
received his M.Sc. in Physics in 1997, his M.Sc. in electronic Engineering in 2003 and his Ph.D. in telecommunication engineering (Dr. Ing.) in 2006. He is currently Associate Professor in the Polytechnic University of Valencia, Spain. He is the head of the research group ¡±Communications and Networks¡± of the Integrated Management Coastal Research Institute. He has been Internet Technical Committee chair (IEEE Communications Society and Internet society) for the term 2013-2015. He has authored 22 book chapters and has more than 360 research papers published in national and international conferences, international journals (more than 140 with ISI Thomson JCR). He has been the co-editor of 40 conference proceedings and guest editor of several international books and journals. He is editor-in-chief of the ¡±Ad Hoc and Sensor Wireless Networks¡± (with ISI Thomson Impact Factor), and he is (or has been) associate editor of 46 international journals (16 of them with ISI Thomson Impact Factor). He has been involved in more than 320 Program committees of international conferences, and more than 130 organization and steering committees. He leads many national and international projects. He is currently the chair of the Working Group of the Standard IEEE 1907.1. He has been general chair (or co-chair) of 36 International workshops and conferences. He is IEEE Senior and IARIA Fellow.
\end{IEEEbiography}

\begin{IEEEbiography}[{\includegraphics[width=1in,height=1.125in,clip,keepaspectratio]{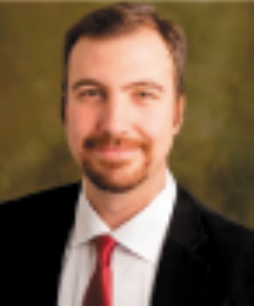}}]{Burak Kantarci}
(SUM’05,M’09,SM’12) is an Assistant Professor with the School of Electrical Engineering and Computer Science at the University of Ottawa. From 2014 to 2016, he was an assistant professor at the ECE Department at Clarkson University, where he currently holds a courtesy appointment. Dr. Kantarci received the M.Sc. and Ph.D. degrees in computer engineering from Istanbul Technical University, in 2005 and 2009, respectively. He received the Siemens Excellence Award in 2005 for his studies in optical burst switching. During his Ph.D. study, he studied as a Visiting Scholar with the University of Ottawa, where he completed the major content of his thesis. He has co-authored over 150 papers in established journals and conferences, and contributed to 12 book chapters. He is the Co-Editor of the book entitled Communication Infrastructures for Cloud Computing. He has served as the Technical Program Co-Chair of seven international conferences/symposia/workshops. He is an Editor of the IEEE Communications Surveys and Tutorials. He also serves as the Vice-Chair of the IEEE ComSoc Communication Systems Integration and Modeling Technical Committee. He is a member of the ACM and a senior member of the IEEE.
\end{IEEEbiography}

\begin{IEEEbiography}[{\includegraphics[width=1in,height=1.125in,clip,keepaspectratio]{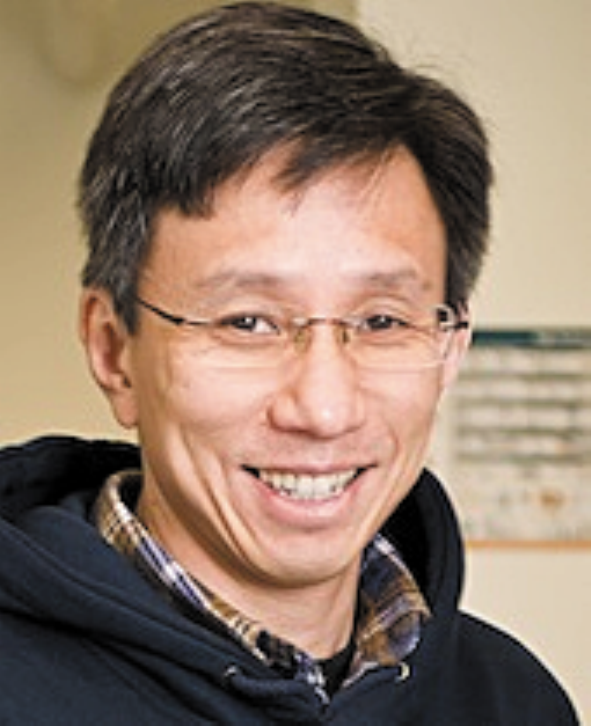}}]{Winston K.G. Seah}
received the Dr.Eng. degree from Kyoto University, Kyoto, Japan, in 1997. He is currently Professor of Network Engineering in the School of Engineering and Computer Science, Victoria University of Wellington, New Zealand. Prior to this, he has worked for more than 16 years in mission-oriented industrial research, taking ideas from theory to prototypes, most recently, as a Senior Scientist (Networking Protocols) in the Institute for Infocomm Research (I$^2$R), Singapore. He is actively involved in research in the areas of mobile ad hoc and sensor networks, and co-developed one of the first Quality of Service (QoS) models for mobile ad hoc networks. His latest research is focused on networking protocols to address the needs of 5G networks, the Internet of Things and other machine-type communications (MTC) technologies, encompassing both long-range communications (LTE-A, Narrowband IoT) as well as, short range technologies (IEEE802.15.4, 6LoWPAN, RPL, etc.) He is a Senior Member of the IEEE and Professional Member of the ACM. His detailed CV is available at http://www.ecs.vuw.ac.nz/~winston/.
\end{IEEEbiography}

\end{document}